\documentclass[11pt]{article}

\usepackage{template}
\usepackage{natbib}
\usepackage{accents}
\usepackage{algcompatible}
\usepackage{multirow}
\usepackage{setspace}

\title{Spatially Varying Gene Regulatory Networks via Bayesian Nonparametric Covariate-Dependent Directed Cyclic Graphical Models}

\makeatletter
\renewcommand\@fnsymbol[1]{}
\makeatother

\author{
  Trisha Dawn$^{*}$,
  Yang Ni$^{\dagger}$\thanks{Email: trisha@stat.tamu.edu,  yang.ni@austin.utexas.edu} \\
  {\normalsize $^{*}$ Department of Statistics, Texas A\&M University} \\
  {\normalsize $^{\dagger}$ Department of Statistics and Data Sciences, University of Texas at Austin}
}

\date{}

\begin{document}

\maketitle

\begin{abstract}
    \noindent  Spatial transcriptomics technologies enable the measurement of gene expression with spatial context, providing opportunities to understand how gene regulatory networks vary across tissue regions. However, existing graphical models focus primarily on undirected graphs or directed acyclic graphs, limiting their ability to capture feedback loops that are prevalent in gene regulation. Moreover, ensuring the so-called stability condition of cyclic graphs, while allowing graph structures to vary continuously with spatial covariates, presents significant statistical and computational challenges. We propose BNP-DCGx, a Bayesian nonparametric approach for learning spatially varying gene regulatory networks via covariate-dependent directed cyclic graphical models. Our method introduces a covariate-dependent random partition as an intermediary layer in a hierarchical model, which discretizes the covariate space into clusters with cluster-specific stable directed cyclic graphs. Through partition averaging, we obtain smoothly varying graph structures over space while maintaining theoretical guarantees of stability. We develop an efficient parallel tempered Markov chain Monte Carlo algorithm for posterior inference and demonstrate through simulations that our method accurately recovers both piecewise constant and continuously varying graph structures. Application to spatial transcriptomics data from human dorsolateral prefrontal cortex reveals spatially varying regulatory networks with feedback loops, identifies potential cell subtypes within established cell types based on distinct regulatory mechanisms, and provides new insights into spatial organization of gene regulation in brain tissue.
   \vspace{5pt}
    
   \noindent \textbf{Keywords:} Directed cyclic graphs; Covariate dependent partition model; Parallel tempering; MCMC; Covariate-dependent graphs  
\end{abstract}

\section{Introduction}

Single-cell RNA sequencing (scRNA-seq) has revolutionized our understanding of cellular heterogeneity and gene expression by enabling the characterization of distinct cell types within complex tissues. Despite its widespread applications in fields such as cancer research, immunology, and neuroscience, scRNA-seq alone lacks spatial context, which is crucial for addressing many biological questions. Spatial transcriptomics bridges this gap by integrating scRNA-seq with spatial localization within tissues. This spatial information provides critical insights into gene expression patterns and spatially regulated functions. For instance, brain function is closely tied to the spatial organization of neuronal and glial cells. Spatial transcriptomics techniques can be broadly classified into four categories: (i) sequencing-based (e.g., 10x Genomics Visium, Slide-seq, Stereo-seq, Light-seq); (ii) probe-based (e.g., GeoMx); (iii) imaging-based (e.g., CosMx SMI, MERFISH, seqFISH, seqFISH+); and (iv) image-guided spatially resolved scRNA-seq (e.g., Geo-seq, Zipseq).
 For an overview of spatial transcriptomics technologies, we refer to \cite{chen2023spatial}.

Building on advancements in spatial transcriptomics, a critical step in analyzing such data is to understand spatially varying gene regulatory networks (svGRNs), which capture how gene regulatory interactions change across tissue regions. Unlike traditional gene regulatory networks (GRNs), which assume uniform regulatory interactions, svGRNs account for spatial heterogeneity by incorporating cellular location and local microenvironmental influences. 
The existence of spatially varying genes, whose expression levels change across different spatial regions, and the existence of cellular heterogeneity are indications of svGRNs, as these variations in expression and cellular functions often arise from localized regulatory interactions \citep{wang2022high}. Various computational approaches have been developed to identify spatially varying genes, including model-based methods (e.g., SpatialDE \citep{svensson2018spatialde}, SPARK \citep{sun2020statistical}, SPARK-X \citep{zhu2021spark}), neighborhood-based methods (e.g., Giotto \citep{dries2021giotto}, nnSVG \citep{weber2023nnsvg}, MERINGUE \citep{miller2021characterizing}), and statistical measures such as Moran’s I \citep{moran1950notes}, as implemented in the \texttt{Seurat} package \citep{hao2023, hao2021, stuart2019, butler2018, satija2015}. While identifying spatially varying genes has been extensively studied, uncovering svGRNs is underexplored and could provide deeper insights into spatially regulated biological processes, developmental patterns, and disease mechanisms.  

Our motivating application involves analyzing spatial transcriptomics data from the human dorsolateral prefrontal cortex \citep{maynard2021transcriptome}, where cells exhibit distinct spatial organization across tissue layers. For instance, Figure \ref{fig: SpatialHeatmap} displays the spatial distribution of gene expression levels for eight selected genes from neurons, revealing clear spatial patterns and heterogeneity. Traditional GRN analyses assume uniform regulatory relationships across all cells of a given type. However, the spatial organization of brain tissue suggests that gene regulatory interactions may vary across different tissue regions, even within the same cell type. Our goal is to infer svGRNs that capture how regulatory relationships change with cellular location, potentially revealing distinct regulatory mechanisms in different microenvironments. Such analysis requires methods capable of learning directed cyclic graphs to capture feedback loops inherent in gene regulation while accounting for spatial heterogeneity.

\begin{figure}[h!]
    \centering
        \includegraphics[width = 15cm, height = 6.5cm]{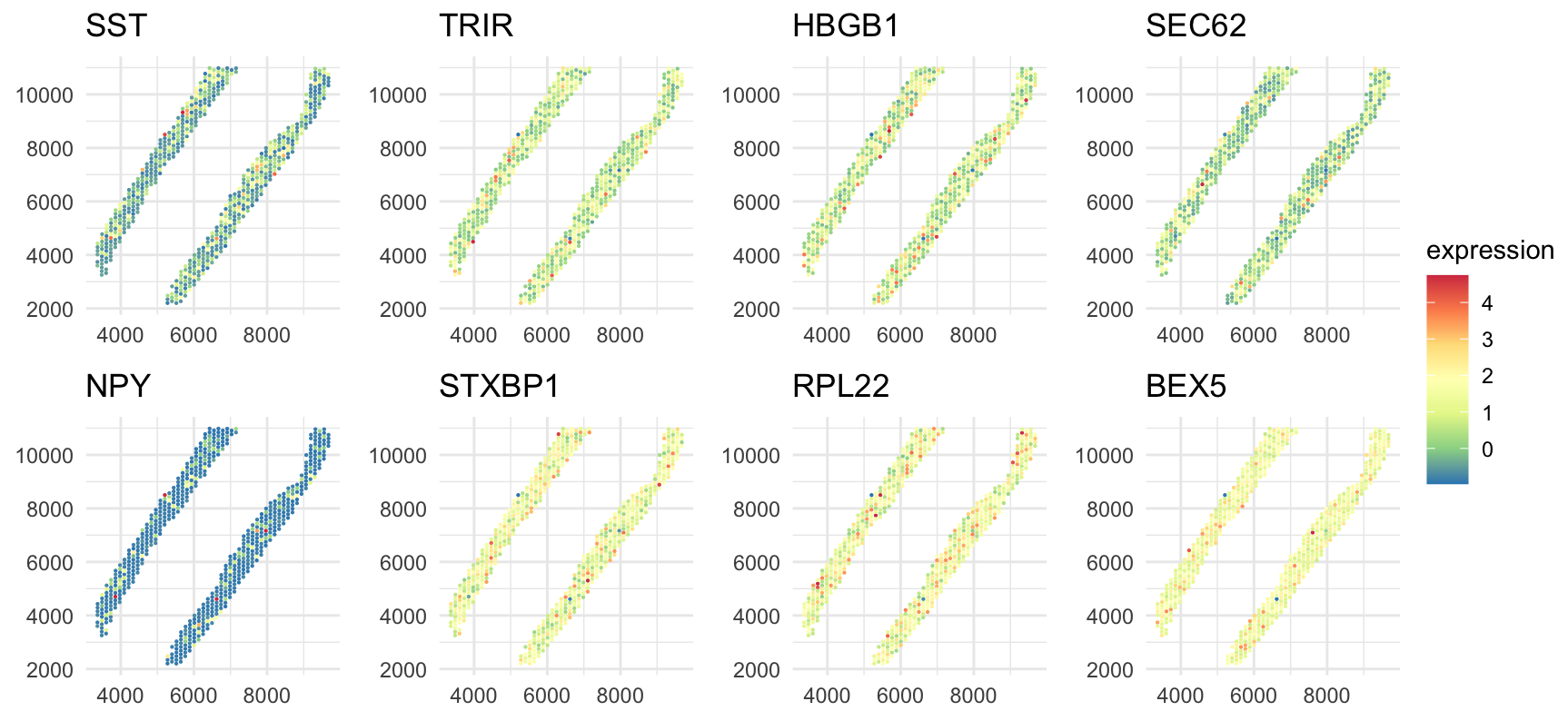}
	\caption{Spatial distribution of the gene expression levels for the $8$ selected genes from neurons category.}
		\label{fig: SpatialHeatmap}
\end{figure}


The construction of GRNs from gene expression data while incorporating spatial cell location information aligns closely with broader statistical research on graph learning from heterogeneous data. This area has gained significant attention due to its applications in various biological domains, such as cancer genomics, where molecular and genomic heterogeneity is prevalent \citep{heppner1983tumor}. Much of the existing work in this field focuses on leveraging auxiliary information to recover graph structures and has been dedicated to learning either covariate-assisted undirected graphs (UGs) or directed acyclic graphs (DAGs), limiting their ability to capture feedback loops and spatial variations in gene regulation.

Research on UG learning is extensive, with various studies exploring covariate-dependent Gaussian graphical models (GGMs) to infer graph structures while adjusting for covariates on the mean level \citep{yin2011sparse, li2012sparse, cai2013covariate, bhadra2013joint, JMLR:v17:16-004, chen2016asymptotically, ni2018reciprocal, deshpande2019simultaneous, niu2020bayesian}. Some approaches have considered modeling covariate-dependent precision matrices
\citep{zhou2010time, fox2015bayesian, lee2018nonparametric, xie2020identifying,  zeng2024bayesian}. While these methods improve UG structure learning in heterogeneous data, they often assume a single homogeneous UG structure across populations. To address this limitation, researchers have developed approaches for learning multiple graphs in grouped settings \citep{guo2011joint,danaher2014joint,peterson2015bayesian,mitra2016bayesian,xie2016joint,tan2017bayesian,lin2017joint,shaddox2018bayesian} and mixtures of GGMs to estimate cluster-specific graphs \citep{rodriguez2011sparse,talluri2014bayesian}. Other works have introduced regression-based or partition-based methods to allow graph structure to vary with covariates \citep{liu2010graph,ni2022bayesian,wang2022bayesian,niu2023covariate,zhang2023high,chen2023probabilistic,yao2025robust}.
In the bioinformatics literature, there are some attempts to infer gene co-expression networks from spatial transcriptomic data \citep{li2025spagrn,fang2023subcellular,detomaso2021hotspot,wei2022spatial,acharyya2022spacex}, but they too mostly rely on UGs.

DAG learning has also been widely explored, with efforts to account for population heterogeneity in grouped settings \citep{oates2014joint, yajima2015detecting}. More recent work has developed methods for continuously varying DAG structures \citep{ni2019bayesian,zhou2023individualized} and covariate-dependent quantile DAG estimation \citep{sagar2022bayesian}. However, despite the great advancements, these approaches either impose strict acyclicity constraints or fail to capture the inherent asymmetry in gene regulation, limiting their applicability in constructing GRNs with possible feedback loops.

 To address these challenges, some studies have explored learning directed cyclic graphs (DCGs), which accommodate feedback loops. Notable approaches include the cyclic causal discovery (CCD) algorithm \citep{richardson1996discovery}, the independent component analysis-based LiNG-D algorithm \citep{lacerda2012discovering}, and the cyclic additive noise models \citep{mooij2011causal}. However, these methods were not designed to accommodate covariates or model non-iid/heterogeneous data. \cite{ni2018reciprocal,ni2018heterogeneous} incorporated covariates but only on the mean level, and hence the graph remains the same for all observations. 
Unlike UG and DAG learning, one unique challenge of inferring DCGs whose structures vary with covariates is that the \textit{stability condition}, which is an eigenvalue condition on the causal effect matrix (detailed in Section \ref{Section 2}), has to be satisfied at any given covariate value or spatial location.

 In this article, we introduce a novel method, called Bayesian nonparametric directed cyclic graph with covariates (BNP-DCGx), to infer directed cyclic svGRNs from spatial transcriptomic data. To accommodate distinct graph structures influenced by spatial variation, we introduce a covariate-dependent random partition layer within a Bayesian hierarchical model, serving as an intermediary between covariates and DCGs. This hierarchical learning framework learns finite partition-dependent DCGs that satisfy the stability condition and then refines them through \textit{partition averaging}, effectively capturing gene regulatory interactions continuously varying over space.  This is, to our knowledge, the first approach that allows for the modeling of svGRNs with feedback loops while accounting for spatial heterogeneity.

The rest of the article is organized as follows. Section \ref{Section 2} introduces the notations, general formulation, and the proposed model, detailing the likelihood, prior distributions, and random partition-based cyclic graphical models. Section \ref{Section 3} discusses posterior inference using a parallel tempered Markov chain Monte Carlo (MCMC) sampler, practical implementation of partition averaging, and the theoretical guarantee of satisfying the stability condition for partition-specific DCGs.  Section \ref{Theory} establishes theoretical support for the proposed method. 
Section \ref{Section 4} provides simulations, including cluster-specific graphs
and continuously varying covariate-dependent DCGs. Finally, Section \ref{Section 5} 
applies the proposed method to the motivating spatial transcriptomics data to infer svGRNs, followed by a discussion in Section \ref{Section 6}.

\section{Proposed Model}\label{Section 2}

\subsection{Notations, Backgrounds, and General Formulation}\label{sec:nbg}

Let $Y_i =(y_{i1}, \cdots, y_{ip})^\top \in \mathbbm{R}^p$ denote a vector of $p$ gene expressions for cell $i=1,\dots,n$. 
Let $X_i= (x_{i1}, \cdots, x_{iq})^\top \in \mathbbm{R}^q$ denote a vector of covariates, which are 2-dimensional spatial locations in our application (hence $q=2$). For any positive integer $n$, let $[n]=\{1,\dots,n\}$ and let $A_{[n]}$ be a shorthand for $A_1,\dots,A_n$.

Traditionally, constructing GRNs ignores $X_i$'s and is only based on $Y_i$'s, which are assumed to be iid. Specifically, a structural equation model \citep[SEM]{bollen1989structural} is often used,
\begin{align}\label{eq:sem0}
    Y_i = M + BY_i + E_i,
\end{align}
where $M\in \mathbbm{R}^p$ is a vector of intercepts, $B\in \mathbbm{R}^{p\times p}$ is a matrix of SEM coefficients representing gene regulatory effects in our application, and $E_i\in \mathbbm{R}^p$ is a vector of jointly independent noises. To connect SEM with a graph, we draw an arrow $j\to k$ meaning gene $j$ regulates gene $k$ if $B_{kj}\neq 0$. It can be shown that the SEM in \eqref{eq:sem0} is globally Markov with respect to the resulting graph, i.e., it respects all the conditional independence relationships encoded in the graph via the notion of d-separation \citep{spirtes1995directed}. If the graph is constrained to be acyclic, the SEM is a Bayesian network model \citep{peal1985bayesian}. However, generally, the graph need not be acyclic. For instance, if $B_{jk}\neq 0$ and $B_{kj}\neq 0$, then $j\rightleftarrows k$. For historic reasons, directed graphs without acyclic constraints are called DCGs even though they are not required to contain directed cycles. The diagonal entries of $B$ are often fixed to 0 to avoid self-loops. DCGs and SEMs have been widely used in engineering \citep{mason1953feedback} and econometrics \citep{haavelmo1943statistical} to represent dynamic processes in equilibrium with the assumption of stability \citep{fisher1970correspondence}. An SEM is \textit{stable} if the modulus of the eigenvalues of $B$ is less than or equal to 1, and no eigenvalue is 1. Under the stability condition, the SEM and its associated DCG can be viewed as a compact representation of an infinite DAG of variables indexed
by time \citep{spirtes1995directed}. 

In this paper, we will relax the iid assumption and allow the graph to vary with covariates (spatial locations) by extending  \eqref{eq:sem0} to
\begin{align}\label{eq:sem}
    Y_i = M(X_i) + B(X_i)Y_i + E_i,
\end{align}
where $M(\cdot):\mathbbm{R}^q\mapsto\mathbbm{R}^{p}$ and $B(\cdot):\mathbbm{R}^q\mapsto\mathbbm{R}^{p\times p}$ map covariates to the SEM intercepts and coefficients. Because $B(\cdot)$ encodes the graph structure, the graph can vary with covariates in our model. 
\textit{The main challenge is how to model the mapping $B(\cdot)$ so that $B(X)$ satisfies the stability condition for any $X\in\mathbbm{R}^q$.} 
\begin{definition}[Stable Function]
   We say a matrix-variate function $B(\cdot):\mathbbm{R}^q\mapsto\mathbbm{R}^{p\times p}$ is stable if $B(X)$ is stable for any $X\in\mathbbm{R}^q$
\end{definition}

Natural modeling choices such as splines and Gaussian processes would find it quite hard to ensure the stability of $B(\cdot)$. For instance, one might attempt to expand $B_{jk}(X)=\sum_{\ell=1}^L \beta_{jk\ell}\phi_{jk\ell}(X)$ for some basis functions $\{\phi_{jk\ell}(\cdot)\}$. Then it seems unattainable to introduce and implement a set of constraints on $\{\phi_{jk\ell}(\cdot)\}$ or $\{\beta_{jk\ell}\}$ to force all the eigenvalues of $B(X)$ for any $X$ to have modulus less than or equal to 1.

To overcome this challenge, we propose a novel Bayesian nonparametric covariate-dependent DCG model, BNP-DCGx, which is illustrated in Figure \ref{fig: dualview}. The main idea is to not contemplate explicit expression of $B(X)$ but instead introduce an intermediary between $B(\cdot)$ and $X$ in a Bayesian hierarchical formulation, which is designed to enable a simple strategy for imposing the stability constraints while maintaining functional flexibility. We choose the intermediary to be a covariate-dependent infinite partition, which has some desirable properties. First, conditional on the partition, $B(\cdot)$ is discretized into $B_1,B_2,\dots$, each of which complies with the stability condition straightforwardly. Second, marginalizing out the partition, we recover $B(X)$, which is guaranteed to be stable for any $X$. Hence, our model can be viewed from two different perspectives, namely, the cluster view and the functional view. Third, the partition size (i.e., the number of clusters) grows with sample size, which gives rise to a flexible representation of $B(\cdot)$. The general hierarchical formulation is as follows (for simplicity, some parameters such as the intercepts are omitted),
\begin{align}\label{eq:clusterv}
    \mathbbm{P}(Y_{[n]},S_{[L_n]}, B_{[L_n]}| X_{[n]})=\overbrace{\mathbbm{P}\left(Y_{[n]}|S_{[L_n]}, B_{[L_n]}  \right)}^{\text{Sampling distribution}}\underbrace{\mathbbm{P}\left(B_{[L_n]}\right)}_{\substack{\text{Prior distribution}\\\text{with stability constraint}}}\overbrace{\mathbbm{P}\left(S_{[L_n]}|X_{[n]}\right)}^{\substack{\text{Covariate-dependent}\\ \text{partition prior}}},
\end{align}
where $S_{[L_n]}=\{S_1,\dots,S_{L_n}\}$ is a partition of $[n]$. We call $S_\ell$ a cluster, and the number of clusters $L_n$ is allowed to grow with $n$.
A key feature of this formulation is that $B(\cdot)$ is no longer explicitly expressed as in \eqref{eq:sem}. It is replaced by countably many $B_\ell$'s for which the prior distributions, to be specified later, can easily accommodate the stability constraint. Because $B_\ell$'s are only defined conditionally on the clusters/partition, we call \eqref{eq:clusterv} the cluster view of the proposed model. 

Although $B(\cdot)$ is not explicit in  \eqref{eq:clusterv}, it is implicitly defined as
\begin{align*}
&[B(X_1),\dots,B(X_n)]|S_{[ L_n]}\; {\buildrel\rm d\over=} \;[B_{\xi_1},\dots,B_{\xi_n}],\\
&S_{[ L_n]}|X_{[n]}\sim \mathbbm{P}\left(S_{[L_n]}\middle|X_{[n]}\right),
\end{align*}
where $\xi_i$ is the cluster indicator of unit $i$ with $\xi_i=\ell$ if $i\in S_\ell$ for $\ell=1,\dots,L_n$. Marginalizing out $S_{[L_n]}$, we arrive at the functional view of the proposed model -- a mixture distribution of $B(X_1),\dots,B(X_n)$,
\begin{align*}
    \mathbbm{P}(B(X_1),\dots,B(X_n))=\sum_{S_{[ L_n]}\in \Pi_n}\mathbbm{P}(S_{[ L_n]}|X_{[n]})\mathbbm{P}(B_{[L_n]})\prod_{i=1}^nI(B(X_i)=B_{\xi_i})
\end{align*}
where $\Pi_n$ is the collection of all partitions of $[n]$ and $I(\cdot)$ is the indicator function. We will show that 
\begin{align*}
    \mathbbm{P}(B(X_i)~\text{is stable}~\forall i)=1.
\end{align*}

Denote $\mathbb{S}_p$ to be the set of all $p \times p$ stable matrices. Then the following holds,
\begin{align*}
    \begin{split}
        \mathbbm{P}\left( B(X_i) \in \mathbb{S}_p~\forall i\right) & = \mathbb{E}_{S_{[L_n]}}\big[\mathbbm{P}\left( B(X_i) \in \mathbb{S}_p ~\forall i \big | S_{[L_n]}\right)\big]\\
        & = \mathbb{E}_{S_{[L_n]}}\big[\mathbbm{P}\left( B_{\xi_i} \in \mathbb{S}_p ~\forall i \big | S_{[L_n]}\right)\big] \\
        & = 1.
    \end{split}
\end{align*}

The second equality holds because conditional on a random partition $S_{[L_n]}$, the cluster indicators $\xi_i$ are non-random quantities and $B(X_i) = B_{\xi_i}$. Moreover, each of the $B_{\xi_i}$ satisfy the stability condition.

For the rest of this section, we will specify the sampling distribution, priors on $\{B_\ell\}$ that respects the stability condition, the covariate-dependent partition prior, and priors of other parameters omitted in \eqref{eq:clusterv}.
\begin{figure}[h!]
    \centering
        \includegraphics[width = 13cm]{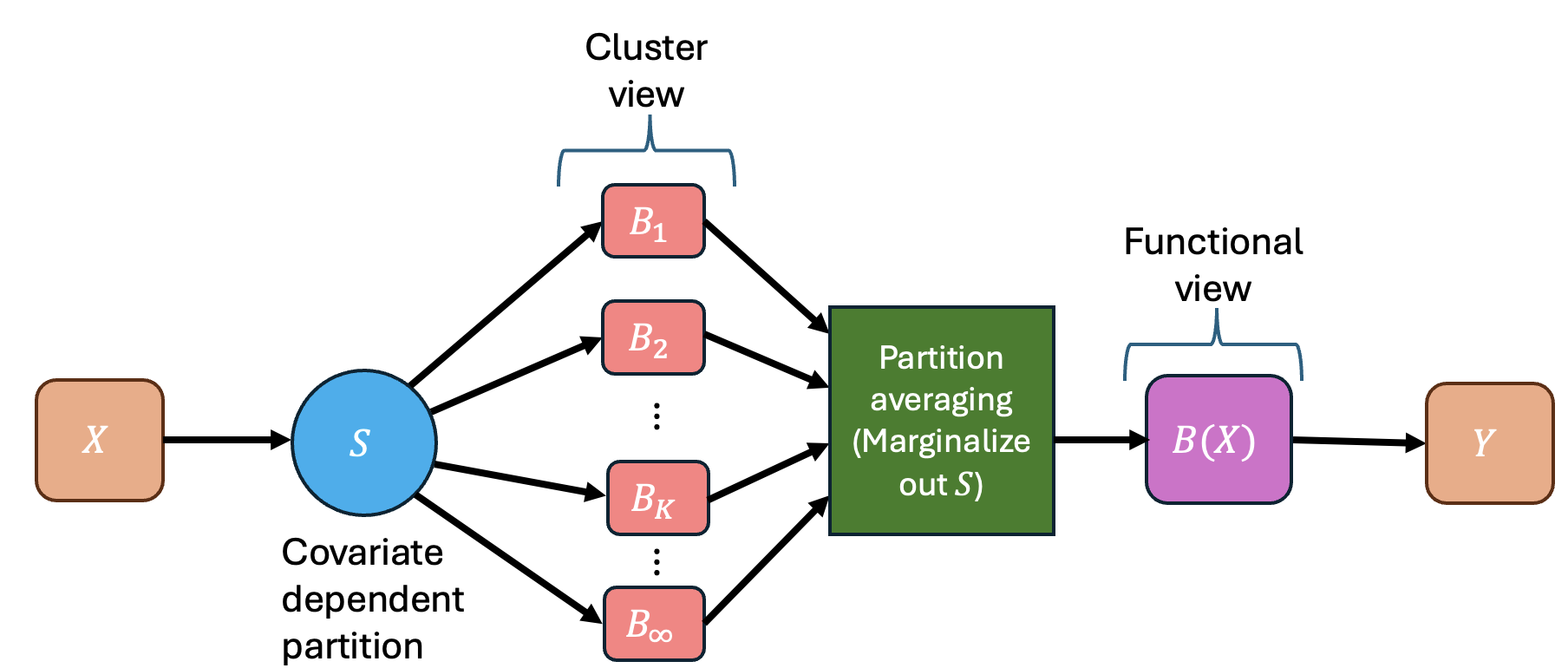}
       
	\caption{Schematic representation of BNP-DCG with dual views: cluster view and functional view. For simplicity, all the model parameters are not included in the figure. }
		\label{fig: dualview}
\end{figure}

\subsection{Sampling Distribution}
To specify the sampling distribution of $Y_i$'s is to specify the distribution $\mathbbm{P}_E$ of the noises $E_i$'s. Although Gaussian would be a tempting choice, it should be avoided for identifiability. Gaussian noises would render the same sampling distribution for Markov equivalent DCGs (i.e., DCGs that encode the same set of conditional independence relationships), and hence those DCGs cannot be distinguished from each other even with an infinite sample size. Fortunately, Gaussian noises are a peculiar case in that any non-Gaussian noises would resolve the identifiability issue thanks to the identifiability theory of independent component analysis \citep{lacerda2012discovering}. Therefore, in this paper, we choose noises to be Laplace-distributed,
\begin{align*}
    E_i|\xi_i=\ell\overset{ind}{\sim}\mathbbm{P}_E(E_i)=\prod_{j=1}^p\text{Laplace}\left(E_{ij}\middle|\sqrt{\frac{\sigma_{\ell j}}{2}}\right),
\end{align*}
for $i=1,\dots,n$ and $j=1,\dots,p$. For ease of computation, we equivalently represent the Laplace distribution as a scale mixture of normal,
\begin{align*}
    &E_i|\xi_i=\ell\overset{iid}{\sim} N_p(0, \Sigma_{\ell i}), \\
    &\Sigma_{\ell i}=\text{diag}(\sigma_{\ell 1}\tau_{i1}, \sigma_{\ell 2}\tau_{i2}, \cdots, \sigma_{\ell p}\tau_{ip}),\\
    &\tau_{ij}\overset{iid}{\sim} \text{Exp}(1).
\end{align*}

Then via the change of variables, the sampling distribution of $Y_i$'s is given by,
\begin{align*}
    \mathbbm{P}(Y_{[n]}|S_{[ L_n]}, \{ B_\ell,M_\ell,\sigma_{\ell1},\dots,\sigma_{\ell p}\}_{\ell=1}^{L_n})=\prod_{i=1}^n\mathbbm{P}_E\left((\mathbbm{I}_p - B_{\xi_i})Y_i-M_{\xi_i}\right)\text{det}(\mathbbm{I}_p - B_{\xi_i})
\end{align*}
where $\mathbbm{I}_p$ is a $p$-dimensional identity matrix and $\text{det}(\cdot)$ is matrix determinant.

\subsection{Covariate-Dependent Partition Prior}

The partition $S_{[ L_n]}$ is the crucial intermediary between $B_{[L_n]}$ and $X_{[L_n]}$, which ultimately leads to a stable $B(\cdot)$.
We consider the product partition model with covariates \citep[PPMx]{muller2011product},
\begin{align}\label{eg:PPMX}
\mathbbm{P}(S_{[ L_n]}|X_{[n]})\propto \prod_{\ell = 1}^{L_n} g(X_{S_\ell})c(S_\ell),
\end{align}
where $X_{S_\ell}=\{X_i|i\in S_\ell\}$, $g(X_{S_\ell})\geq 0$ is a similarity function characterizing the similarity among $X_{S_\ell}$, and $c(S_\ell) \geq 0$ is the cohesion function capturing the tightness of $S_\ell$. For ease of computation, we follow the recipe in \cite{muller2011product} by setting,
\begin{align}
  \label{eq:similarity}  &g(X_{S_\ell}) = \int \prod_{i \in S_\ell} \mathbbm{P}(X_i |\Psi_\ell)\mathbbm{P}(\Psi_\ell)\, d \Psi_\ell,\\
 \label{eq:cohesion}
    &c(S_\ell) = \alpha(|S_\ell| - 1)!,
\end{align}
where $|S_\ell|$ is the size of cluster $\ell$. These specifications of the similarity and the cohesion functions would later allow us to reformulate our full model with augmented parameters $\Psi_{[L_n]}$ as a Dirichlet process mixture model 
for which many existing posterior sampling algorithms can be employed.

The auxiliary probability models $\mathbbm{P}(X_i |\Psi_\ell)$ and $\mathbbm{P}(\Psi_\ell)$ in \eqref{eq:similarity} are given by,
\begin{align}
    \nonumber &X_i|\Psi_\ell \sim N_q(\mu_\ell, \Lambda_\ell),\\ 
   \label{eg:probmodelx} &\Psi_\ell=(\mu_\ell,\Lambda_\ell)\sim N_q(0, \omega \Lambda_\ell)\times \text{IW}(q, \mathbbm{I}_q),
\end{align}
where $\text{IW}(q, \mathbbm{I}_q)$ is the inverse-Wishart distribution with  degrees of freedom $q$ and an identity scale matrix $\mathbbm{I}_q$.

\subsection{Stable
Spike-and-Slab and Cascade Priors}\label{sec:sssc}

Conditional on the partition, we assign priors to the cluster-specific parameters $\{B_\ell,M_\ell, \sigma_{\ell1},\dots,\sigma_{\ell p}\}_{\ell=1}^{L_n}$.
The SEM coefficient matrix $B_\ell$ needs to be both sparse and stable. To achieve this, we propose a stable spike-and-slab (SSS) prior,
\begin{equation}
    \begin{aligned}\label{eq:sss}
    \mathbbm{P}(B_\ell|\gamma_\ell,\eta_\ell)&=\frac{1}{C_1(\gamma_\ell,\eta_\ell)} I(B_\ell\in\mathbb{S}_p)\prod_{j\neq k} \left\{\gamma_{\ell jk}N(B_{\ell jk}|0,\eta_\ell)+(1-\gamma_{\ell jk})N(B_{\ell jk}|0,\nu_0\eta_\ell)\right \},\\
    &=\frac{1}{C_1(\gamma_\ell,\eta_\ell)} I(B_\ell\in\mathbb{S}_p)\prod_{j\neq k} N_{p(p-1)}(\vec{B}_\ell|0,V_\ell),
\end{aligned}
\end{equation}
where $\mathbb{S}_p$ is the collection of all $p\times p$ stable matrices, $\gamma_{\ell jk}\in\{0,1\}$, $\nu_0$ is a very small fixed constant, $\vec{B}_\ell$ is a $p(p-1)$-dimensional vector created by stacking the columns of $B_\ell$ without the diagonal entries, $V_\ell=\eta_\ell\cdot\text{diag}(\vec{v}_\ell)$, $\vec{v}_\ell=[v_{\ell jk}]$, $v_{\ell jk}=\gamma_{\ell jk}+(1-\gamma_{\ell jk})\nu_0$, and $C(\gamma_\ell,\eta_\ell)$ is a normalizing factor,
\begin{align*}
C_1(\gamma_\ell,\eta_\ell)=\int     I(B_\ell\in\mathbb{S}_p)\prod_{j\neq k} \left\{\gamma_{\ell jk}N(B_{\ell jk}|0,\eta_\ell)+(1-\gamma_{\ell jk})N(B_{\ell jk}|0,\nu_0\eta_\ell)\right \} dB_\ell.
\end{align*}
The binary indicator $\gamma_{\ell jk}$ indicates whether $k\to j$ is present in cluster $\ell$ because the spike $N(B_{\ell jk}|0,\nu_0\eta_\ell)$ is very tightly concentrated on 0 whereas the slab $N(B_{\ell jk}|0,\eta_\ell)$ is dispersed. 

Because the SSS prior enforces stability via $I(B_\ell\in\mathbb{S}_p)$, the normalizing factor $C_1(\gamma_\ell,\eta_\ell)$ is intractable. To avoid evaluating it in the calculation of the full conditional distributions of $(\gamma_\ell,\eta_\ell)$, we assume the following joint prior on $(\gamma_\ell,\eta_\ell)$,
\begin{align}\label{eq:gammaeta}
    \mathbbm{P}(\gamma_\ell,\eta_\ell|\phi_\ell)=\frac{1}{C_2(\phi_\ell)}C_1(\gamma_\ell,\eta_\ell)\text{IG}(\eta_\ell|a_\eta,b_\eta)\prod_{j\neq k}\text{Bernoulli}(\gamma_{\ell jk}|\phi_\ell),
\end{align}
where ${IG}(\cdot|a_\eta,b_\eta)$ is the inverse-gamma distribution and the normalizing factor is given by,
\begin{align*}
    C_2(\phi_\ell)=\int\int C_1(\gamma_\ell,\eta_\ell)\text{IG}(\eta_\ell|a_\eta,b_\eta)\prod_{j\neq k}\text{Bernoulli}(\gamma_{\ell jk}|\phi_\ell)d\gamma_\ell d\eta_\ell.
\end{align*}
Including $C_1(\gamma_\ell,\eta_\ell)$ in $\mathbbm{P}(\gamma_\ell,\eta_\ell|\phi_\ell)$ serves to cancel out its reciprocal in $\mathbbm{P}(B_\ell|\gamma_\ell,\eta_\ell)$ when deriving the full conditional distributions for $\gamma_\ell$ and $\eta_\ell$. With this specification, their full conditionals are Bernoulli and inverse-gamma.  In the same vein, we include $C_2(\phi_\ell)$ in the prior of $\phi_\ell$ to have closed-form full conditional,
\begin{align}\label{eq:rho}
    \mathbbm{P}(\phi_\ell)=\frac{1}{C_3}C_2(\phi_\ell)\text{beta}(\phi_\ell|a_\phi,b_\phi),
\end{align}
where 
\begin{align*}
    C_3=\int C_2(\phi_\ell)\text{beta}(\phi_\ell|a_\phi,b_\phi)d\phi_\ell.
\end{align*}
Although none of the normalizing factors $C_1(\gamma_\ell,\eta_\ell),C_2(\phi_\ell)$, $C_3$ is analytically available, they are all finite (e.g., $C_1$ is bounded by 1 because it is an integral of a probability density function over a subspace) and never have to be evaluated in our posterior inference algorithm. Because $C_1$ cascades through $C_3$ via $C_2$, we call the joint prior \eqref{eq:sss}-\eqref{eq:rho} as the cascade prior.

We complete our model with conjugate priors on the intercepts and the noise scale parameters,
\begin{align*}
&M_\ell \sim N_p(0, \lambda \mathbbm{I}_p),\\
&\sigma_{\ell j} \sim \text{IG}(a_\sigma,b_\sigma).
\end{align*}

\subsection{Partition Averaging }\label{sec:pa}

As mentioned in Section \ref{sec:nbg}, the proposed BNP-DCGx has dual views. Conditional on the partition, the model identifies cluster-specific DCGs, whereas with the partition marginalized out, the functional view is manifested,
\begin{align*}
    &\mathbbm{P}(B(X_1),\dots,B(X_n))=\sum_{S_{[L_n]}\in \Pi_n}\mathbbm{P}(S_{[n]}|X_{[n]})\mathbbm{P}(B_{[L_n]})\prod_{i=1}^nI(B(X_i)=B_{\xi_i})\\
    =&\sum_{S_{[L_n]}\in \Pi_n}\mathbbm{P}(S_{[L_n]}|X_{[n]})\times \prod_{\ell=1}^{L_n}\frac{1}{C_1(\gamma_\ell,\eta_\ell)} I(B_\ell\in\mathbb{S}_p)\prod_{j\neq k} N_{p(p-1)}(\vec{B}_\ell|0,V_\ell)\times \prod_{i=1}^nI(B(X_i)=B_{\xi_i})
\end{align*}
where $\mathbbm{P}(S_{[L_n]}|X_{[n]})$ is defined in \eqref{eg:PPMX}-\eqref{eq:cohesion}. 

Although complicated in form, $\mathbbm{P}(B(X_1),\dots,B(X_n))$ is a discrete mixture of truncated, degenerated multivariate Gaussians with mixture weights $\mathbbm{P}(S_{[L_n]}|X_{[n]})$. The truncation is due to the stability constraint $B_\ell\in\mathbb{S}_p$, and the degeneracy $B(X_i)=B_{\xi_i}$
 reflects the piecewise constant nature of the representation of $B(X)$ via partition.

For an unseen observation $n+1$, the posterior predictive distribution of $B(X_{n+1})$ is given by
\begin{align}\label{eq:ppa}
\begin{split}
    & \mathbbm{P}(B(X_{n+1})|Y_{[n]}, X_{[n]}, X_{n+1})\\& = \sum_{S_{[n]} \in \Pi_n} \sum_{\ell= 1}^{L_n+1} \mathbbm{P}\left(S_{[L_n]}^{(\ell)}|Y_{[n]}, X_{[n]},X_{n+1}\right) \mathbbm{P}(B_\ell|S_{[L_n]}, Y_{[n]}, X_{[n]})I(B(X_{n+1})=B_\ell)
\end{split}
\end{align}
where $S_{[L_n]}^{(\ell)}=S_{[L_n]}\backslash S_\ell\cup\{S_\ell\cup \{n+1\}\}$, which is a partition of $[n+1]$ obtained by adding unit $n+1$ to cluster $S_\ell$ in $S_{[L_n]}$.
We show that the stability condition is preserved with probability one.

\begin{proposition} 
Under the proposed model, $\mathbbm{Pr}(B(X_{n+1})\in \mathbb{S}_p|Y_{[n]},X_{[n]},X_{n+1})=1$.
\end{proposition}
\begin{proof} 
\begin{align*}
        & \mathbbm{Pr} \left( B(X_{n+1})\in \mathbb{S}_p | Y_{[n]}, X_{[n]}, X_{n+1} \right) \\
         =& \mathbbm{E}_{\xi_{n+1}} \left[ \mathbbm{Pr} \left( B(X_{n+1}) \in \mathbb{S}_p | Y_{[n]}, X_{[n]}, X_{n+1},\xi_{n+1} \right)  \right] \\
         =& \mathbbm{E}_{\xi_{n+1}} \left[ \mathbbm{Pr} \left( B_{\xi_{n+1}}\in \mathbb{S}_p | Y_{[n]}, X_{[n]}, X_{n+1},\xi_{n+1} \right)  \right] \\
        =&\mathbbm{E}_{\xi_{n+1}}[1]\\
        =& 1
\end{align*}
The second equality is because conditional on $\xi_i$, $B(X_i)=B_{\xi_i}$ for any $i$. The third equality is because the SSS prior proposed in Section \ref{sec:sssc} ensures $B_\ell\in \mathbb{S}_p$ for any $\ell$.
\end{proof}

\section{Posterior Inference}\label{Section 3}

\subsection{Markov Chain Monte Carlo}\label{MCMC}

The proposed model is parameterized by $M_\ell,B_\ell,\sigma_{\ell j}, \mu_\ell,\Lambda_\ell,\gamma_\ell,\eta_\ell,\phi_\ell,\tau_{ij},\xi_i$ for $\ell\in[L_n]$, $i\in [n]$, and $j\in [p]$. All but $B_\ell$ and $\xi_i$ is sampled via closed-form Gibbs. Metropolis-within-Gibbs is used for sampling $B_\ell$ and $\xi_i$. 
All sampling steps are provided below.

\begin{itemize}
    \item Sample $\mathbf{\phi_\ell}\sim \text{beta}\left(\sum_{j \neq k}\gamma_{\ell jk} + a_\phi, p^2 - p - \sum_{j \neq k}\gamma_{\ell jk} + b_\phi\right)$.

    \item Sample $\sigma_{\ell j}\sim \text{IG}\left(a_\sigma + \frac{|S_\ell|}{2}, 
    b_\sigma + \frac {1}{2} \sum_{i:\xi_i=\ell} \frac{r^2_{ij}}{\tau_{ij}}\right)$ where $r_{ij}=y_{ij}-\sum_{k\neq j}B_{\ell jk}y_{ik}-M_{\ell j}$.
    
    \item Sample $M_\ell\sim N_p\left(\widetilde V \left(\sum_{i:\xi_i=\ell} Y_i^\top (I - B_\ell)^\top \Sigma_{\ell I)}^{-1}\right)^\top, \widetilde V\right)$ where $\widetilde V = \left(\frac{1}{\lambda} \mathbbm{I}_p +  \sum_{i:\xi_i=\ell} \Sigma_{\ell i}^{-1}\right)^{-1}$.

    \item Sample $\gamma_{\ell jk}\sim \text{Bernoulli}(\widetilde \theta)$ where 
    $\widetilde \theta = \frac{\phi_\ell N\left(B_{\ell jk}|0,  \eta_\ell\right)}{\phi_\ell N\left(B_{\ell jk}|0,  \eta_\ell\right) + (1 - \phi_\ell) N\left(B_{\ell jk}|0,\nu_0\eta_\ell\right)}$.

    \item Sample $\eta_\ell\sim \text{IG}\left(a_\eta + \frac{1}{2}(p^2 - p), b_\eta + \frac{1}{2}\sum_{j\neq k}\widetilde B_{\ell jk}^2\right)$ where $\widetilde B_{\ell jk}=\gamma_{\ell jk}B_{\ell jk}+(1-\gamma_{\ell jk})B_{\ell jk}/\sqrt{\nu_0}$.

    \item Sample $B_{\ell jk}$ by random-walk Metroloplis-within-Gibbs. We propose $B_{\ell jk}^*\sim N(B_{\ell jk},\tau_{\text{prop}})$. If the resulting $B_\ell^*$ is not stable, reject it; otherwise, accept it with probability
    \begin{align}\label{eq:MHforB}
       \min\Bigg\{1, \frac{\mathbbm{P}\left( Y_{S_\ell} | B^*_\ell, M_\ell, \sigma_\ell \right) }{\mathbbm{P}\left( Y_{S_\ell} | B_\ell, M_\ell, \sigma_\ell \right) } \Bigg\},
    \end{align}
    where $Y_{S_\ell}=\{Y_i|i\in S_\ell\}$, $\sigma_\ell=\{\sigma_{\ell 1},\dots,\sigma_{\ell p}\}$, and $\mathbbm{P}\left( Y_{S_\ell} | B_\ell, M_\ell, \sigma_\ell \right)$ is the marginal likelihood with $\tau_{ij}$'s integrated out, 
    \begin{align}\label{eq:jointLaplace}
        \mathbbm{P}\left( Y_{S_\ell} | B_\ell, M_\ell, \sigma_\ell \right) \propto 
        \text{det}(\mathbbm{I}_p - B_\ell)^{|S_\ell|}
        e^{- \sum_{i:\xi_i=\ell} \sum_{j = 1}^p \sqrt{\frac{2}{\sigma_{\ell j}}} \big|r_{ij}\big| }.
    \end{align}

    \item Sample $\tau_{ij}\sim \text{GIG}\left(2, \frac{r^2_{ij}}{\sigma_{\xi_ij}} , \frac12\right )$ where $\text{GIG}(a,b,c)$ is the generalized inverse Gaussian distribution with density proportional to
    ${\tau_{ij}^{-\frac 12}} e^{-\frac 12 \left( 2 \tau_{ij} + \frac{r^2_{ij}}{\sigma_j \tau_{ij}} \right)} $.

    \item Sample $\xi_i$ from a categorical distribution with probability 
\begin{align}\label{eq:collapsedGibbs}
    \mathbbm{Pr}(\xi_i=\ell | \xi_{-i}, Y_{[n]}, X_{[n]}, B_{[L_n]}, \sigma_{[L_n]}, \tau_i) \propto \mathbbm{Pr}(\xi_i=\ell|\xi_{-i})\mathbbm{P}(Y_i|Y_{-i}^\ell, B_\ell, \Sigma_{\ell i})\mathbbm{P}(X_i|X_{-i}^\ell),
\end{align}
where $\xi_{-i}=\{\xi_a|a\neq i\}$, and $Y_{-i}^\ell=\{Y_a|a\in S_\ell\backslash \{i\}\}$ and $X_{-i}^\ell=\{X_a|a\in S_\ell\backslash \{i\}\}$ are the observations in the $\ell^{\text{th}}$ cluster except for the $i^{\text{th}}$ one. For better mixing, we have marginalized out $M_\ell$ in \eqref{eq:collapsedGibbs}. The prior conditional distribution of $\xi_i$ is the well-known Chinese restaurant process, 
\begin{align}\label{eq:crp}
\mathbbm{Pr}(\xi_i = \ell|\xi_{-i}) = 
\begin{cases}
\frac{N^j_{-i}}{n + \alpha - 1}, \hspace{5pt} \text{if} \hspace{5pt} \ell \hspace{5pt} \text{is one of the existing clusters}  \\
\frac{\alpha}{n + \alpha - 1},  \hspace{5pt} \text{if} \hspace{5pt} \ell \hspace{5pt} \text{is a new cluster}
\end{cases}
\end{align}

where $N^\ell_{-i}=|S^\ell_{-i}|$ with $S^\ell_{-i}=S_\ell\backslash \{i\}$. 
The conditional distribution of $Y_i$ is given by,
\begin{align}\label{eq:postpredY}
    \mathbbm{P}(Y_i|Y_{-i}^\ell, B_\ell, \Sigma_{\ell i}) = 
    \begin{cases}
    \begin{split}
         & N_p\left(\mathcal{M}^Y_\ell, \mathcal{V}^Y_\ell \right), \hspace{5pt} \text{if} \hspace{5pt} \ell \hspace{5pt} \text{is one of the existing clusters} \\ \\
         & N_p\left(0, \mathcal{V}^Y_* \right),  \hspace{5pt} \text{if} \hspace{5pt} \ell \hspace{5pt} \text{is a new cluster}
    \end{split}
    \end{cases}
\end{align}
with
\[\mathcal{M}^Y_\ell = (\mathbbm{I}_p - B_\ell)^{-1} \left[N_{-i}^\ell \mathbbm{I}_p + \frac{1}{\lambda} \Sigma_{\ell i}\right]^{-1}(\mathbbm{I}_p - B_\ell)\sum_{i' \in S^{\ell}_{-i}} Y_{i'}, \]

\[\mathcal{V}^Y_\ell = (\mathbbm{I}_p - B_\ell)^{-1} \Bigg[\mathbbm{I}_p - \left[(N_{-i}^\ell + 1)\mathbbm{I}_p + \frac{1}{\lambda}\Sigma_{\ell i}\right]^{-1} \Bigg]^{-1} (\mathbbm{I}_p - B_\ell)^{-\top},  \]

\[\mathcal{V}^Y_* = (\mathbbm{I}_p - B^*)^{-1} \Bigg[ \mathbbm{I}_p - \left[\mathbbm{I}_p + \frac{1}{\lambda}\Sigma^*\right]^{-1} \Bigg]^{-1} (\mathbbm{I}_p - B^*)^{-\top},\] 
where $\Sigma^*=\text{diag}(\sigma^*_1\tau_{i1},\dots,\sigma^*_p\tau_{ip})$, and $\sigma_j^*$ and $B^*$ are drawn from their respective prior distributions. 
The conditional distributions for $X_i$ is multivariate t distribution,
\begin{align}\label{eq:postpredX}
 \mathbbm{P}(X_i|X_{-i}^\ell) = 
    \begin{cases}
       & t_{N^\ell_{-i} + 1} \left( \frac{\omega}{1 + N^\ell_{-i}\omega} \sum_{i' \in S^{\ell}_{-i}} X_{i'}, \mathcal{V}^X_\ell  \right), \hspace{5pt} \text{if} \hspace{5pt} \ell \hspace{5pt} \text{is one of the existing clusters}\\
       & t_1 \left( 0, (1 + \omega)\mathbbm{I}_q \right), \hspace{5pt} \text{if} \hspace{5pt} \ell \hspace{5pt} \text{is a new cluster}
    \end{cases}
\end{align}

with
\[\mathcal{V}^X_\ell=\frac{\omega + 1 + N^\ell_{-i}\omega}{(N^\ell_{-i} + 1)(1 + N^\ell_{-i} \omega)}\left[\mathbbm{I}_q - \frac{\omega}{1 + N^\ell_{-i}\omega} \left(\sum_{i' \in S^{\ell}_{-i}} X_{i'} \right)\left(\sum_{i' \in S^{\ell}_{-i}} X_{i'} \right)^\top + \sum_{i' \in S^{\ell}_{-i}} X_{i'} X_{i'}^\top \right].\]

Following Algorithm $8$ from \cite{neal2000markov}, we draw $\xi_i^*$ from \eqref{eq:collapsedGibbs}. 

\end{itemize}

For better mixing, we employ the parallel tempering technique \citep{swendsen1986replica,geyer1991markov} by running $W$ Markov chains in parallel with different but related target distributions. The target distribution is proportional to the fractional power of the likelihood multiplied by the prior. In particular for a sequence of increasing  ``temperature" values $1 = T_1 < \cdots < T_{W}$, we have the target distribution as the following, 
\begin{align}\label{eq:tempPosterior}
\begin{split}
   & \left[\prod_{i = 1}^n \mathbbm{P}\left( Y_i|B_{\xi_i}, M_{\xi_i}, \sigma_{\xi_i}, \tau \right)\mathbbm{P}\left(X_i|\mu_{\xi_i}, \Lambda_{\xi_i}\right)\right]^{1/t}  \mathbbm{P}\left(B, M, \gamma, \eta, \theta, \sigma, \tau, \xi, \mu, \Lambda\right),
\end{split}
\end{align} for any $t \in \{T_1, \cdots, T_{W}\}$. Choosing an increasing sequence of ``temperature" values ensures we have our desired posterior distribution for $T_1 = 1$ along with a set of heated chains. This helps the algorithm explore the parameter space more efficiently by using a flattened likelihood, which is less susceptible to being trapped in local modes.

In addition to the update of parameters within each chain (as described in Section \ref{MCMC}), once every certain number of iterations, we propose to swap the states of two chains (say $C_1$ and $C_2$) and accept it according to a Metropolis-Hastings (MH) ratio,
\begin{align}
 \min\left\{ 1,  \frac{\left[\prod_{i = 1}^n \mathbbm{P}\left( Y_i|B^{(C_1)}_{\xi^{(C_1)}_i}, M^{(C_1)}_{\xi^{(C_1)}_i}, \sigma^{(C_1)}_{\xi^{(C_1)}_i} \right) \right]^{\frac{1}{T_{C_2}}} \left[\prod_{i = 1}^n \mathbbm{P}\left( Y_i|B^{(C_2)}_{\xi^{(C_2)}_i}, M^{(C_2)}_{\xi^{(C_2)}_i}, \sigma^{(C_2)}_{\xi^{(C_2)}_i} \right)\right]^{\frac{1}{T_{C_1}}}}
    {\left[\prod_{i = 1}^n \mathbbm{P}\left( Y_i|B^{(C_1)}_{\xi^{(C_1)}_i}, M^{(C_1)}_{\xi^{(C_1)}_i}, \sigma^{(C_1)}_{\xi^{(C_1)}_i} \right) \right]^{\frac{1}{T_{C_1}}} \left[\prod_{i = 1}^n \mathbbm{P}\left( Y_i|B^{(C_2)}_{\xi^{(C_2)}_i}, M^{(C_2)}_{\xi^{(C_2)}_i}, \sigma^{(C_2)}_{\xi^{(C_2)}_i} \right)\right]^{\frac{1}{T_{C_2}}}} \right\},
\end{align} where the superscript indexes the chain (e.g., $\xi^{(C_1)}_i$ is the cluster-specific indicators for the $i^{\text{th}}$ observation obtained from the chain $C_1$). 

\subsection{Implementation of Partition Averaging with Stability Guarantee}\label{StabilityPartition}

The partition averaging described in Section \ref{sec:pa} can be carried out straightforwardly with MCMC samples. 
For instance, for an unseen observation $n+1$, the posterior mean of $B(X_{n+1})$ can be computed as follows:
\begin{align*}
    &\mathbbm{E}\left[B(X_{n+1})|Y_{[n]},X_{[n]},X_{n+1}\right]\\
    =&\int B(X_{n+1})\sum_{S_{[L_n]} \in \Pi_n} \sum_{\ell= 1}^{L_n+1} \mathbbm{P}\left(S_{[L_n]}^{(\ell)}|Y_{[n]}, X_{[n]},X_{n+1}\right) \mathbbm{P}(B_\ell|S_{[L_n]}, Y_{[n]}, X_{[n]})I(B(X_{n+1})=B_\ell)dB(X_{n+1})\\
    \approx & \frac{1}{T}\sum_{t=1}^TB_{\xi_{n+1}^t}^{t},
\end{align*}
where $t$ indexes the MCMC sample and  $\xi_{n+1}^t$ can be drawn from the following conditional distribution, 
\begin{align*}
    \mathbbm{P}(\xi_{n+1} | \xi_{[n]}^t, X_{[n]},X_{n+1}) \propto \mathbbm{P}(\xi_{n+1}|\xi_{[n]}^t)\mathbbm{P}(X_{n+1}|X_{[n]}),
\end{align*}
for which the closed-form expressions are similar to \eqref{eq:crp} and \eqref{eq:postpredX}.

\section{Large prior support}\label{Theory}
Recall that our prior model entails $B(X_i)=B_{\xi_i}$, i.e., $B(X)$ is a piecewise constant function. Despite the discrete nature, we show that our prior places positive mass around any $\nu\mbox{-}$H$\Ddot{o}$lder continuous $B(\cdot)$ and $M(\cdot)$, or in other words, the proposed model is flexible enough for modeling continuously varying DCG.

\begin{theorem}\label{thm:BM}
 Suppose $(B_0(\cdot))_{p \times p}$ be any matrix of SEM coefficients and $(M_0(\cdot))_{p \times 1}$ be the intercept vector
as in Equation~\eqref{eq:sem0},
where $B_0(\cdot)$, $M_0(\cdot)$ are continuous functions for univariate $x \in (a,b]$ with $-\infty<a<b<\infty$. We assume that all the entries of $B_0(\cdot)$ and $M_0(\cdot)$ belong to the space of uniformly $\nu\mbox{-}$H$\Ddot{o}$lder continuous functions. Then for any $\epsilon>0$, there exists $K^* \in \mathbb{N}$ such that 
 \begin{itemize}
  \item  [(i)] For any absolutely continuous prior distributions on $(B_j, \mu_j, \lambda_j) \in (\mathbb{S}_p, \mathbb{R}, \mathbb{R}^{+})$ and on \\
  \noindent $(\pi_1, \cdots, \pi_{K^\ast})$ we have,
\begin{align*}
\begin{split}
     &\mathbb{P}\left(\left\| B_0 -\sum_{j = 1}^{K^*} \pi_j(x)B_j \right\|_1 < \epsilon \right) > 0\\
 \end{split}
\end{align*}
where $\mathbb{S}_p$ be the set of $p \times p$ stable matrices and $\pi_j(x) = \frac{\pi_j N \left(x; \mu_j, \lambda_j \right)}{\sum_{j = 1}^{K^\ast} \pi_j N \left(x; \mu_j, \lambda_j \right)}$ for $j \in \{1, \cdots, K^*\}$. 
 
 \item [(ii)] Similarly, for any absolutely continuous prior distributions on $(M_j, \mu_j, \lambda_j) \in (\mathbb{R}^{p}, \mathbb{R}, \mathbb{R^+})$ and on $(\pi_1, \cdots, \pi_{K^\ast})$ we have,
\begin{align*}
 \begin{split}
     &\mathbb{P}\left(\left\| M_0 -\sum_{j = 1}^{K^*} \pi_j(x)M_j \right\|_1 < \epsilon \right) > 0\\
 \end{split}
\end{align*}
\end{itemize}
where $\pi_j(x)$ is of the same form as part $(i)$ of the theorem.
 \end{theorem}

Refer to Appendix~\ref{Appendix} for the proof of Theorem~\ref{thm:BM}.

\section{Simulation study}\label{Section 4}

We assess the performance of the proposed BNP-DCGx under two different simulation scenarios. The first simulation demonstrates the capability of BNP-DCGx in learning the cluster-specific graph structures (i.e., the cluster view). The other simulation examines the partition averaging for DCGs that vary with covariates (i.e., the functional view) using
continuous function. Throughout the simulations and the real data analysis, we use the same set of hyperparameters: $\lambda = 10,a_\sigma = b_\sigma = 2, a_\phi = b_\phi = 1,a_\eta = b_\eta = 0.01,\nu_0 = 0.01,\omega = 100$.  We run four parallel chains with temperature values $2.5, 2, 1.5, 1$ and swap the chains every $10$ iterations.

\subsection{Cluster-Specific Graphs}

In this scenario, we consider $L_n = 3$ cluster-specific covariate-dependent DCGs with $p = 10$ nodes. The true DCGs are provided in Figure \ref{fig: Sim1} (shown as the red edges) with the edge coefficients generated from $\Delta \cdot U(0.6, 0.8) + (1-\Delta) \cdot U(-0.8, -0.6)$ where $\Delta\sim \text{Ber} (\frac 12)$.  In each cluster, we have $250$ observations, leading to a total of $750$ observations for this experiment. 
The three cluster-specific intercepts $M_1, M_2$ and $M_3$ are generated from $N_{10}(-0.2 \cdot 1_{10}, 10^{-4}\, \mathbbm{I}_{10})$, $N_{10}(0_{10}, 10^{-4}\, \mathbbm{I}_{10})$, and $N_{10}(0.2\cdot1_{10}, 10^{-4}\, \mathbbm{I}_{10})$, respectively. 
The errors $E_i|\xi_i = \ell$ are generated from $N_{10} (0, \Sigma_{\ell i})$ where $\Sigma_{li} = \text{diag}\left(0.1\tau_{i1}, \cdots, 0.1\tau_{ip}\right)$ for $\ell = 1,2,3$ with $\tau_{ij}\sim\operatorname{Exp}(1)$. 
We then let $Y_i = (\mathbbm{I}_p - B_{\ell})^{-1}(M_{\ell} + E_i)$. Finally, the covariates $X_i|\xi_i=\ell$, for $\ell=1,2,3$, are sampled from $N_2(-5 \cdot \mathrm{1}_2, \mathbbm{I}_2)$, $N_2(\mathrm{0}_2, \mathbbm{I}_2)$, and $N_2(5 \cdot \mathrm{1}_2, \mathbbm{I}_2)$, respectively. 

\begin{figure}[h!]
    \centering
        \includegraphics[width = \textwidth]{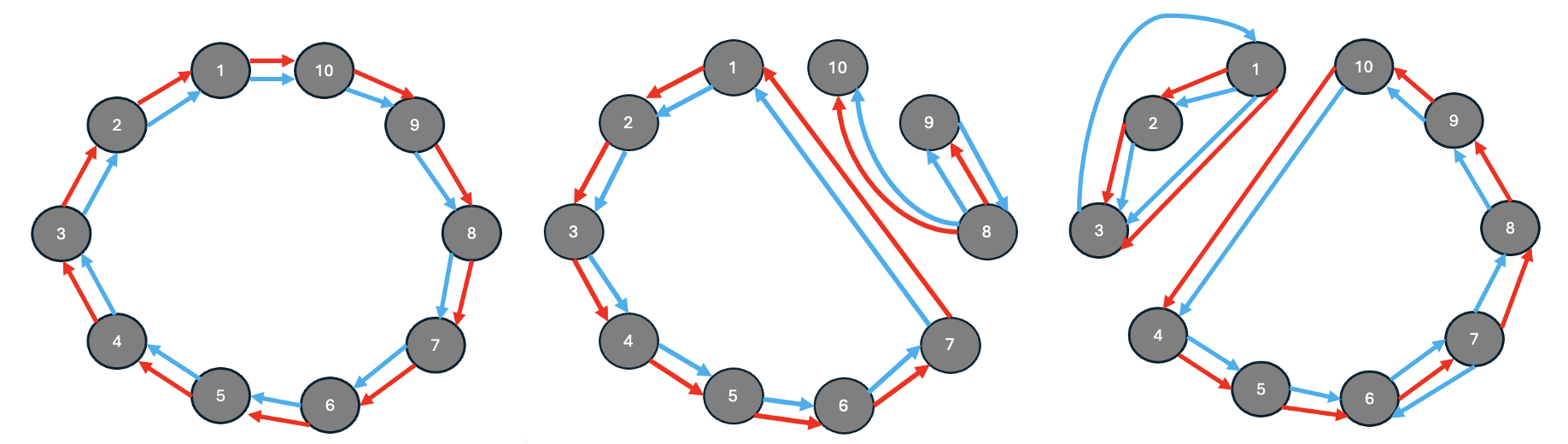}
	\caption{Cluster-specific true (red) DCGs and one illustration of the estimated (blue) DCGs for one particular replication.}
		\label{fig: Sim1}
\end{figure}

 We generate $1000$ MCMC samples from the proposed BNP-DCGx with burn-in of $250$ samples.  We report the estimated cluster-specific DCGs in Figure \ref{fig: Sim1}, shown as blue edges.  
 We compare our method with two alternative methods: (i) the cyclic causal discovery (CCD) algorithm \citep{richardson1996discovery} and (ii) a special case of the proposed method without covariates and clusters (called DCG hereafter). Both methods allow for the modeling of directed cycles but assume that the graph is the same across clusters. 
In Table \ref{ClGrSim1}, we report true positive rate (TPR), false discovery rate (FDR), and Matthew's correlation coefficient (MCC) for all methods calculated over 20 replicates:
 \begin{align*}
    & \text{TPR} = \frac{\text{TP}}{\text{TP} + \text{FN}},~~~~
     \text{FDR} = \frac{\text{FP}}{\text{TP} + \text{FP}},\\
     &\text{MCC} = \frac{(\text{TN}\times \text{TP}) - (\text{FN}\times \text{FP)}}{(\text{TP} + \text{FP})\times(\text{TP} + \text{FN})\times(\text{TN} + \text{FP})\times(\text{TN} + \text{FN})},
 \end{align*}
where TP = number of true positives, FP = number of false positives, TN = number of true negatives, and FN = number of false negatives. It clearly demonstrates the capability of our method in learning cluster-specific DCGs with high TPR and MCC, as well as a low FDR. The competing methods have high FDR, highlighting the importance of accounting for heterogeneity in graph learning when heterogeneity is present.

\begin{table}[ht]
\renewcommand{\arraystretch}{1.4} 
  \begin{tabular}{ |c|c|c|c|c| } 
\hline
 \textbf{Methods} & \textbf{Measures} & \textbf{Cluster} $\mathbf{1}$ & \textbf{Cluster} $\mathbf{2}$ & \textbf{Cluster} $\mathbf{3}$ \\
 \hline
 \multirow{3}{4em}{BNP-DCGx} & TPR & $0.960 (0.027)$ & $0.993 (0.007)$ & $ 1.000 (0.000)$ \\ 
  \cline{2-5}
 & FDR & $0.158 (0.050)$ & $0.148 (0.024)$ & $0.149 (0.023)$\\ 
  \cline{2-5}
 & MCC & $0.882 (0.046)$ & $0.909 (0.015)$ & $0.912 (0.013)$ \\ 
 \hline
 \multirow{3}{4em}{DCG} & TPR & $0.600 (0.093)$ & $0.659 (0.087)$ & $0.707 (0.093)$ \\ 
  \cline{2-5}
 & FDR & $0.734 (0.078)$ & $0.661 (0.090)$ & $0.634 (0.087)$\\ 
  \cline{2-5}
 & MCC & $0.148 (0.030)$ & $0.201 (0.033)$ & $0.244 (0.035)$ \\ 
 \hline
 \multirow{3}{4em}{CCD} & TPR & $0.360 (0.029)$  & $0.274 (0.032)$  &  $0.273 (0.036)$ \\ 
  \cline{2-5}
 & FDR & $0.787 (0.018)$ & $0.855 (0.017)$ & $0.837 (0.022)$ \\ 
  \cline{2-5}
 & MCC & $0.169 (0.027)$ & $0.087 (0.027)$ & $0.093 (0.033)$ \\ 
 \hline
\end{tabular}
\centering
\caption{Average TPR, FDR and MCC along with 
standard errors (in parentheses) for BNP-DCGx, DCG and CCD for all three clusters.}
\label{ClGrSim1}
\end{table}

\subsection{Continuous Graph Functions}

In this scenario, we consider a fixed graph (i.e., one cluster) with continuously varying $B(X)$. For ease of demonstration, we consider $p=3$ and $q=2$. The sample size is set to $n=800$. The true graph is a cycle $1\to2\to3\to1$. The nonzero entries of the true $B(X)$ is a function of covariates $X = (X_1, X_2)$ specified as follows:
$$ f(z) = \frac{e^{3z} - e^{3(1-z)}}{e^{3z} + e^{3(1-z)}} + 0.1 \hspace{5pt} \text{with} \hspace{5pt}  B_{21}(x) = f(x^1_*) , B_{13}(x) =  f(x^2_*) \hspace{5pt} \text{and} \hspace{5pt} B_{32}(x) = f(x^3_*),$$
where $x^1_* = \sqrt{x_1 x_2}, x^2_* = \sqrt{\frac{x^2_1 + x^2_2}{2}}$, and $x^3_* = \frac{x_1 + x_2}{2}$. The rest of the entries of the matrix $B(X)$ is zero. We set $M(X) = (1, 1, 1)^\top$. We generate the noise terms from normal, $E_{ij} \sim N(0, \sigma\tau_{ij})$, with $\sigma 0.1$ and $\tau_{ij}$ drawn from $Exp(1)$ for $ j = 1, 2, 3$. Subsequently, we set $Y_i = (\mathbbm{I}_3 - B(X_i))^{-1}(M + E_{i})$. We uniformly sample the covariates $X_{i1}, X_{i2} \overset{i.i.d.}{\sim} U(0,1)$. Like before, we generate $1000$ MCMC samples with burn-in of $250$ samples.
\begin{figure}[h!]
    \centering
        \includegraphics[width = \textwidth]
        {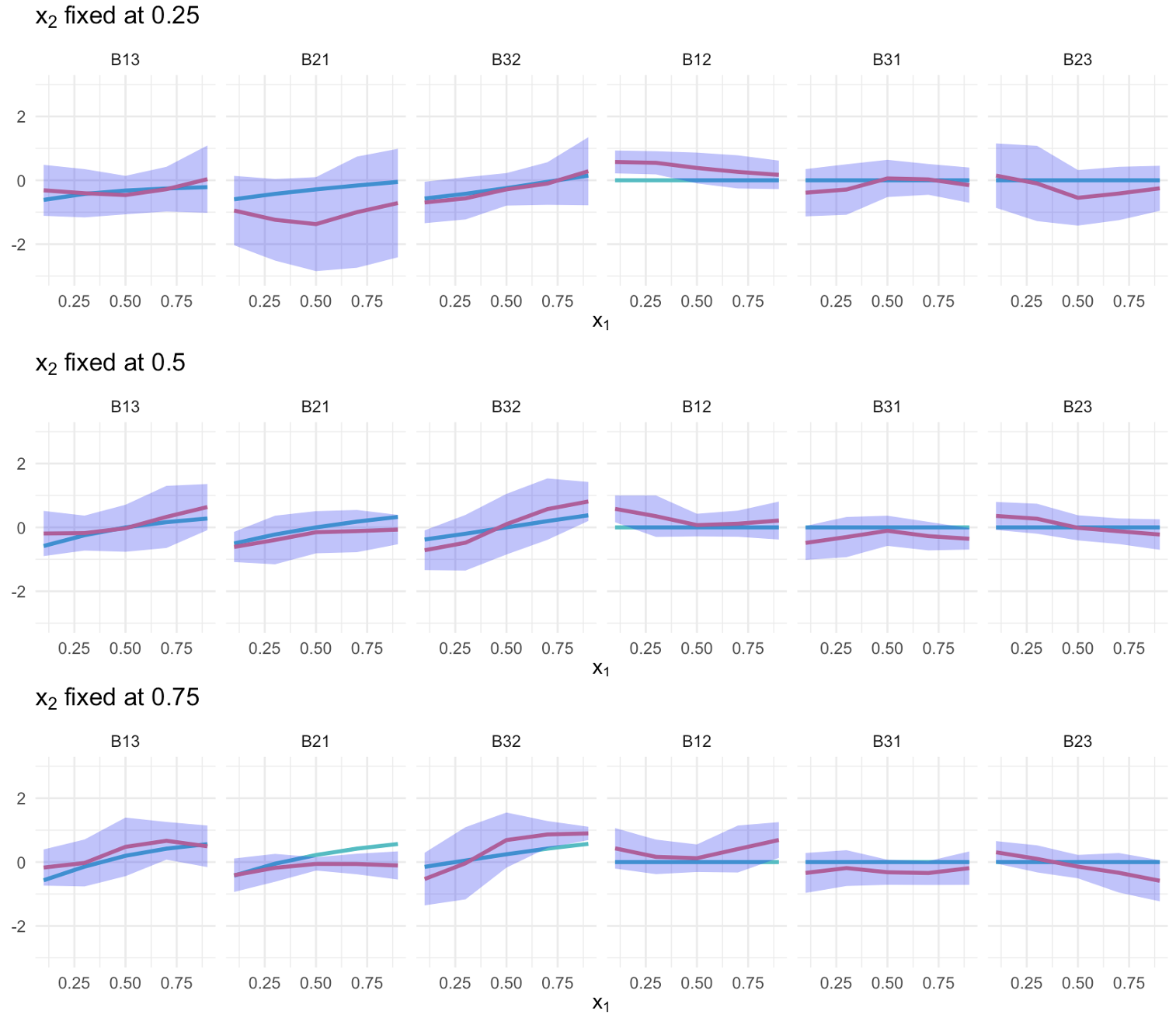}
	\caption{True values of $B(X)$ are represented by the blue curves, the posterior means by the red curves, and the posterior predictive intervals by the shades, for varying $X_1$ and fixed $X_2$.}
		\label{fig: ContGraphVaryX1}
\end{figure}
\begin{figure}[h!]
    \centering
        \includegraphics[width = \textwidth]{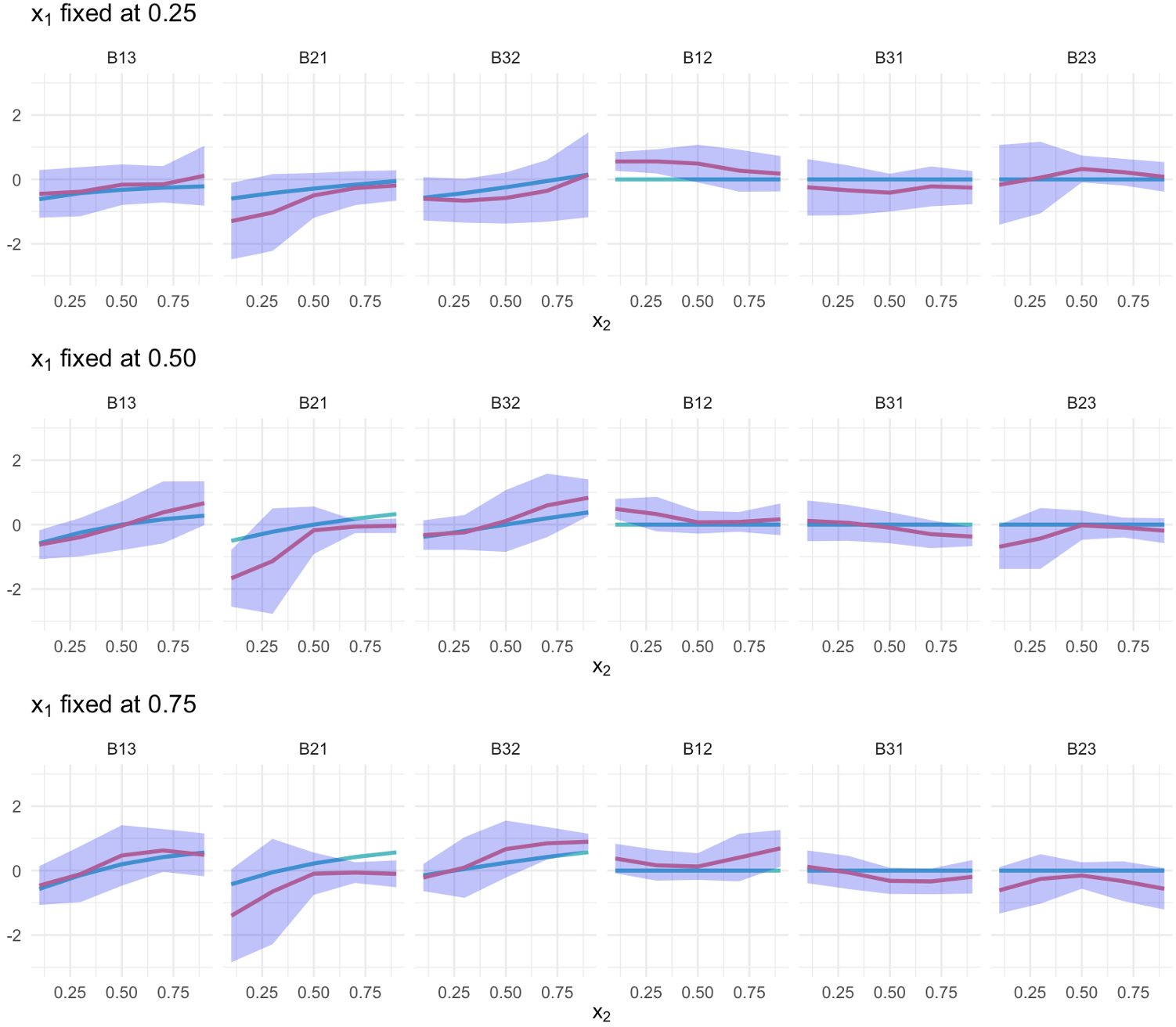}
	\caption{True values of $B(X)$ are represented by the blue curves, the posterior means by the red curves, and the posterior predictive intervals by the shades, for varying $X_2$ and fixed $X_1$.}
		\label{fig: ContGraphVaryX2}
\end{figure}

Figure~\ref{fig: ContGraphVaryX1} (\ref{fig: ContGraphVaryX2}) shows the posterior mean of $B(X)$, along with its predictive intervals (mean $\pm$ 2 standard deviations), as a function of $X_1$ ($X_2$) with $X_2$ ($X_1$) fixed at $0.25, 0.5$, and $0.75$. 
Generally, the true value of $B(X)$ is close to the posterior mean of $B(X)$  and is mostly contained within the limits of the predictive intervals. 

\section{Real Data Analysis}\label{Section 5}

We analyzed the motivating spatial transcriptomics data collected from the six-layered human dorsolateral prefrontal cortex \citep{maynard2021transcriptome} to infer svGRNs. The metadata indicate seven different broad cell types. We arranged our analysis into four tissue segments: excitatory neurons (n=1774), other neurons (n=650), oligodendrocytes and oligodendrocyte-like cells (oligodendrocytes hereafter for short, n=991), and astrocytes (n=1180). 
Excitatory neurons and other neurons exhibit double-band structures (see the left panels of Figures~\ref{fig: union_type3} and \ref{fig: union_type2}),  oligodendrocytes consist of several similar cell subtypes that are spatially adjacent (Figure \ref{fig: union_type4to7}, left panel), and astrocytes form a single band structure (Figure \ref{fig: union_type1}, left panel). 
For objectivity, we selected eight genes with the highest correlations, which is necessary for graph learning, for each tissue segment. These selected genes turn out to be biologically relevant as they are significantly enriched for brain functions (Figure \ref{fig:GO}). We generated $1000$ posterior samples with $250$ burn-in from the proposed BNP-DCGx using the parallel tempered MCMC.

\begin{figure}[h]
    \centering
    \includegraphics[height = 8.5cm,width=\linewidth]
    {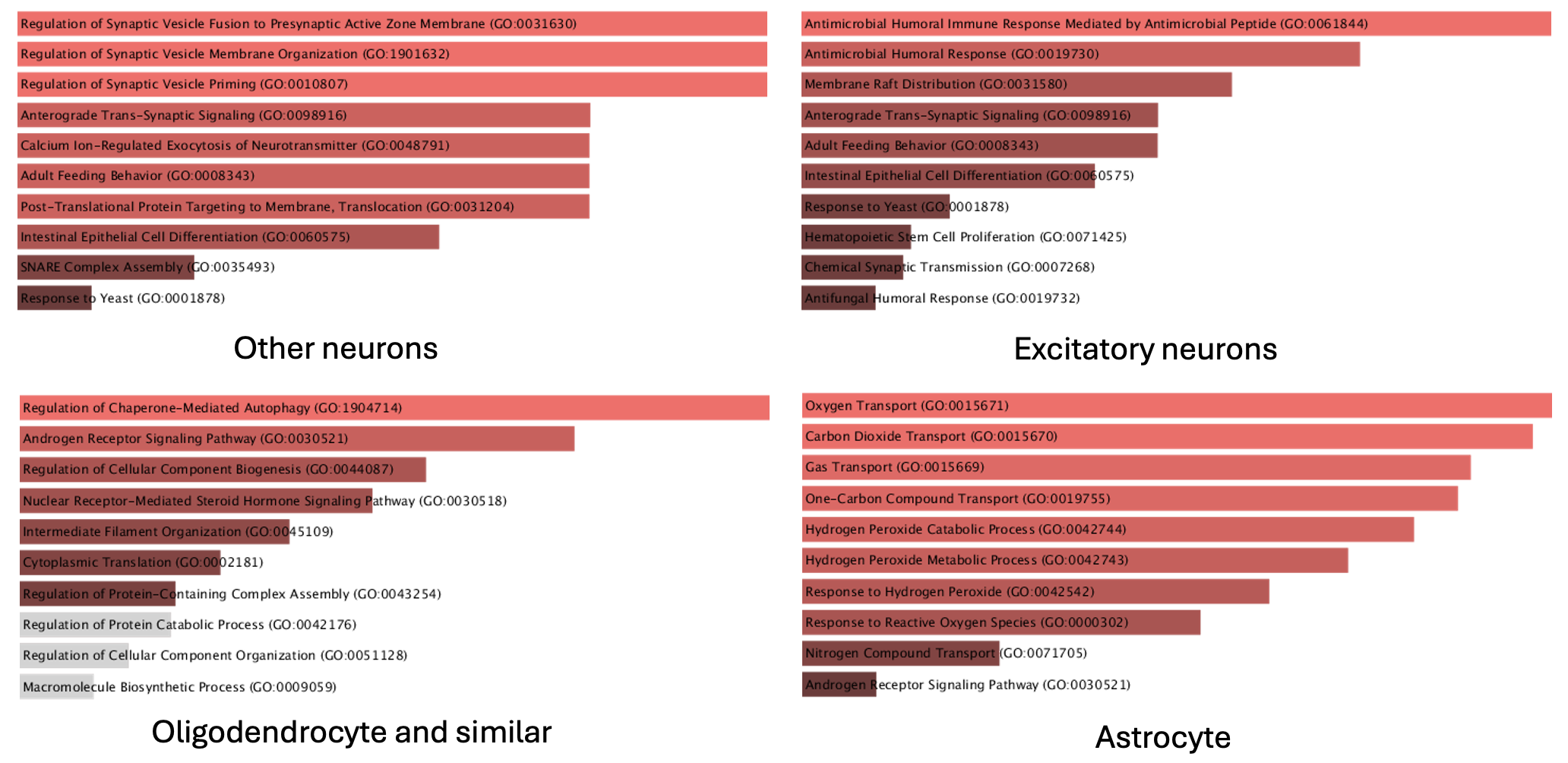}
    \caption{Gene Ontology enrichment analysis of the selected eight genes for each case using Enrichr \citep{kuleshov2016enrichr}. }
    \label{fig:GO}
\end{figure}

\begin{figure}[h!]
    \centering
        \includegraphics[width = 14cm, height = 5cm]{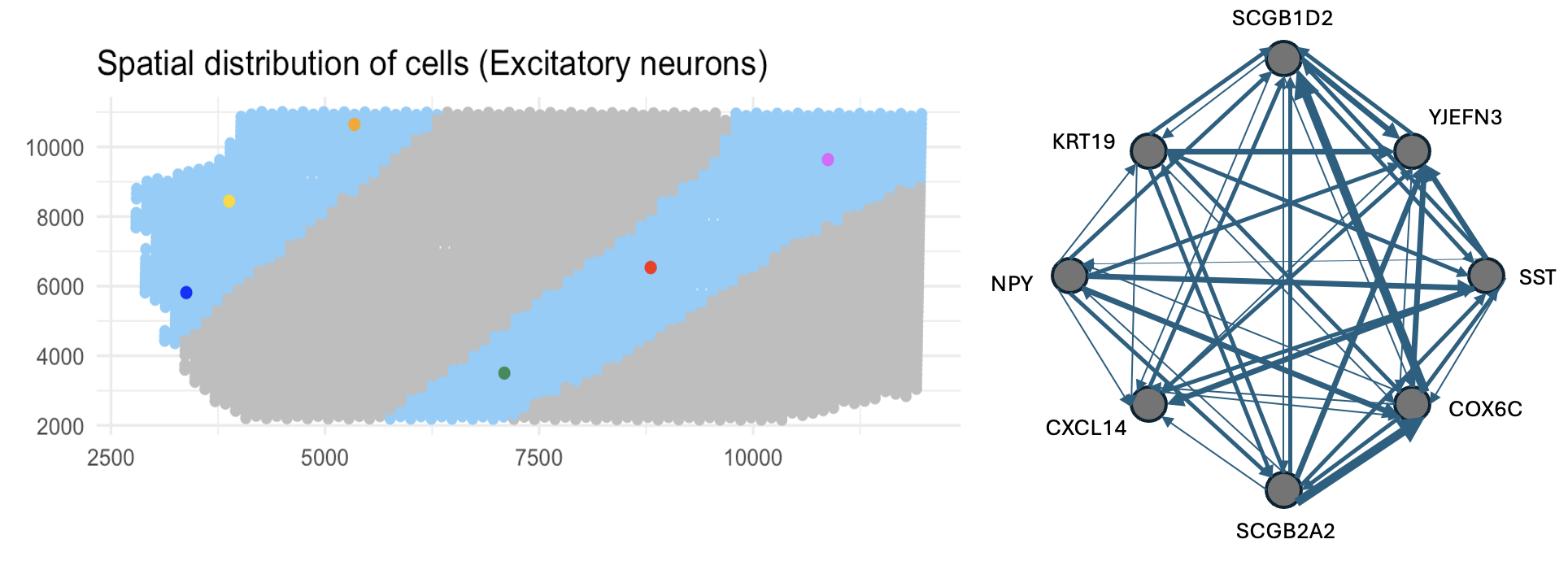}
	\caption{\textit{Left panel:} Spatial distribution of excitatory neurons (in light blue). Six individual spots are highlighted, for which we will show their individual graphs. \textit{Right panel:}  Union graph showing all possible edges at all spot locations. }
		\label{fig: union_type3}
\end{figure}

  \begin{figure}[h!]
    \centering
        \includegraphics[width = 14cm, height = 6cm]
        {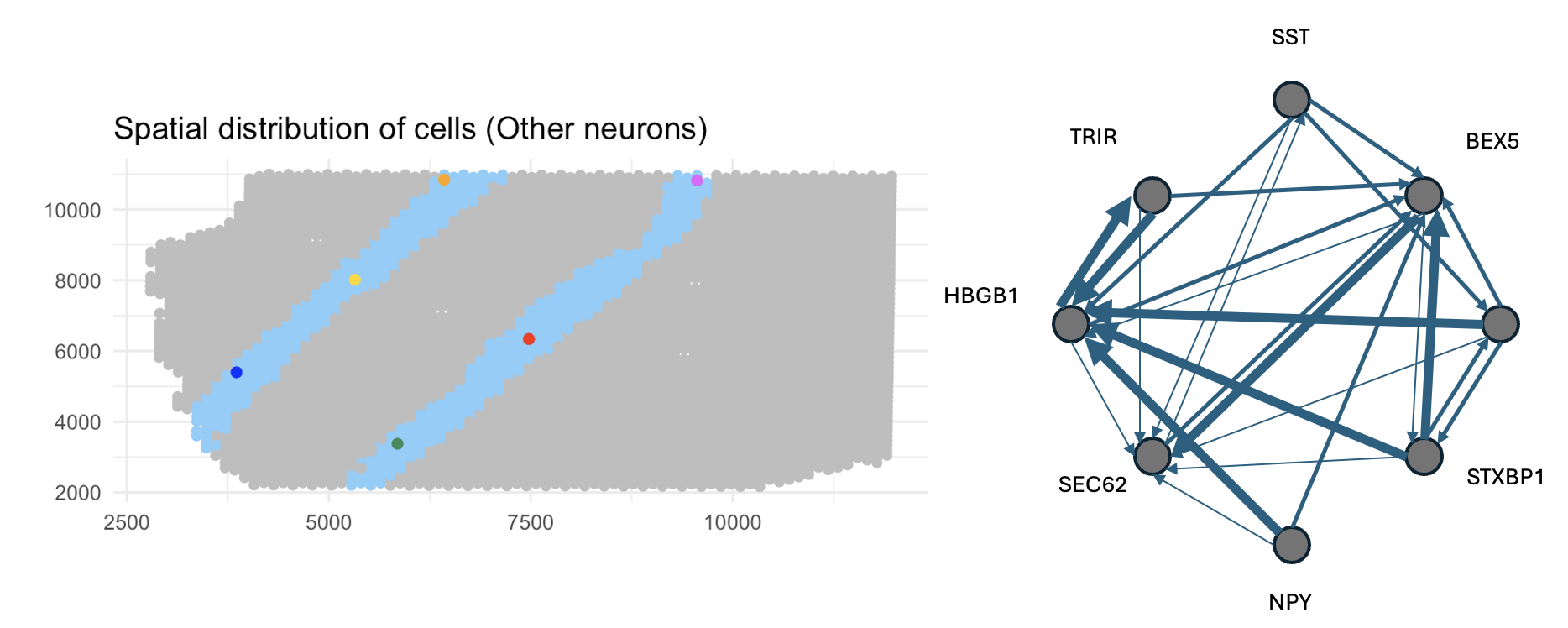}
	\caption{\textit{Left panel:} Spatial distribution of other neurons (in light blue). Six individual spots are highlighted, for which we will show their individual graphs. \textit{Right panel:}  Union graph showing all possible edges at all spot locations. }
		\label{fig: union_type2}
\end{figure}
\begin{figure}[h!]
    \centering
        \includegraphics[width = 14cm, height = 5cm]{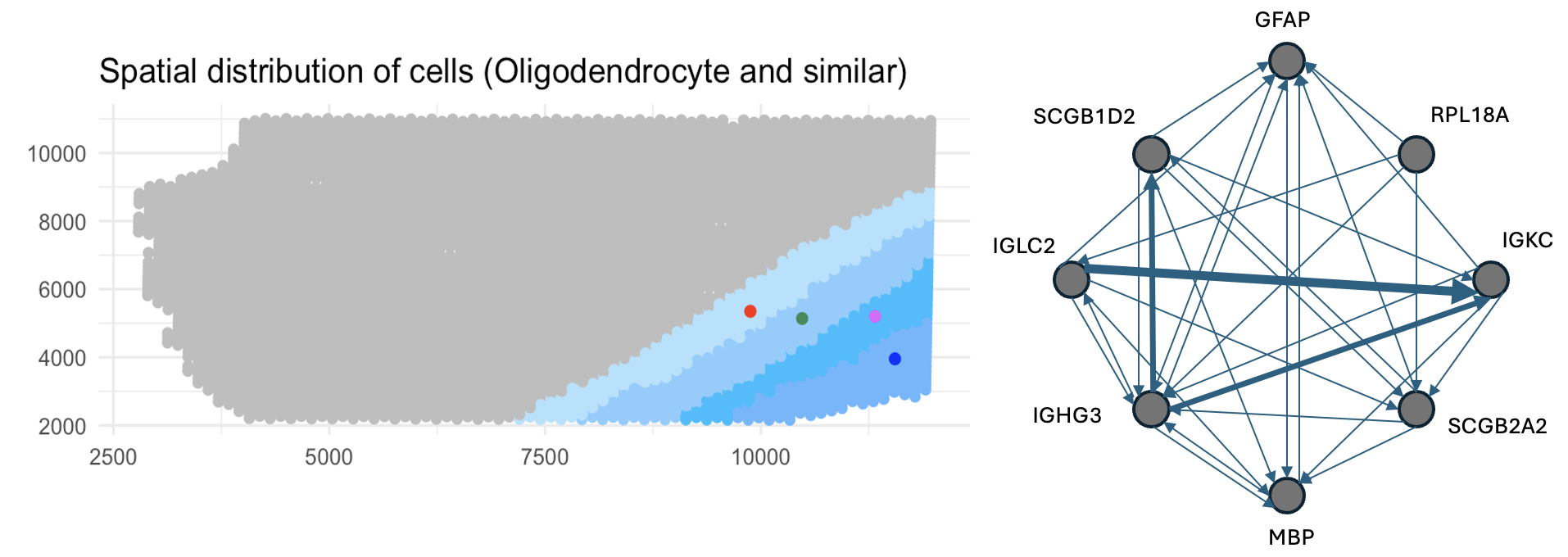}
	\caption{\textit{Left panel:} Spatial distribution of oligodendrocytes (in blue). Four individual spots are highlighted, for which we will show their individual graphs. \textit{Right panel:}  Union graph showing all possible edges at all spot locations. }
		\label{fig: union_type4to7}
\end{figure}
  
    \begin{figure}[h!]
    \centering
        \includegraphics[width = 13cm, height = 5cm]{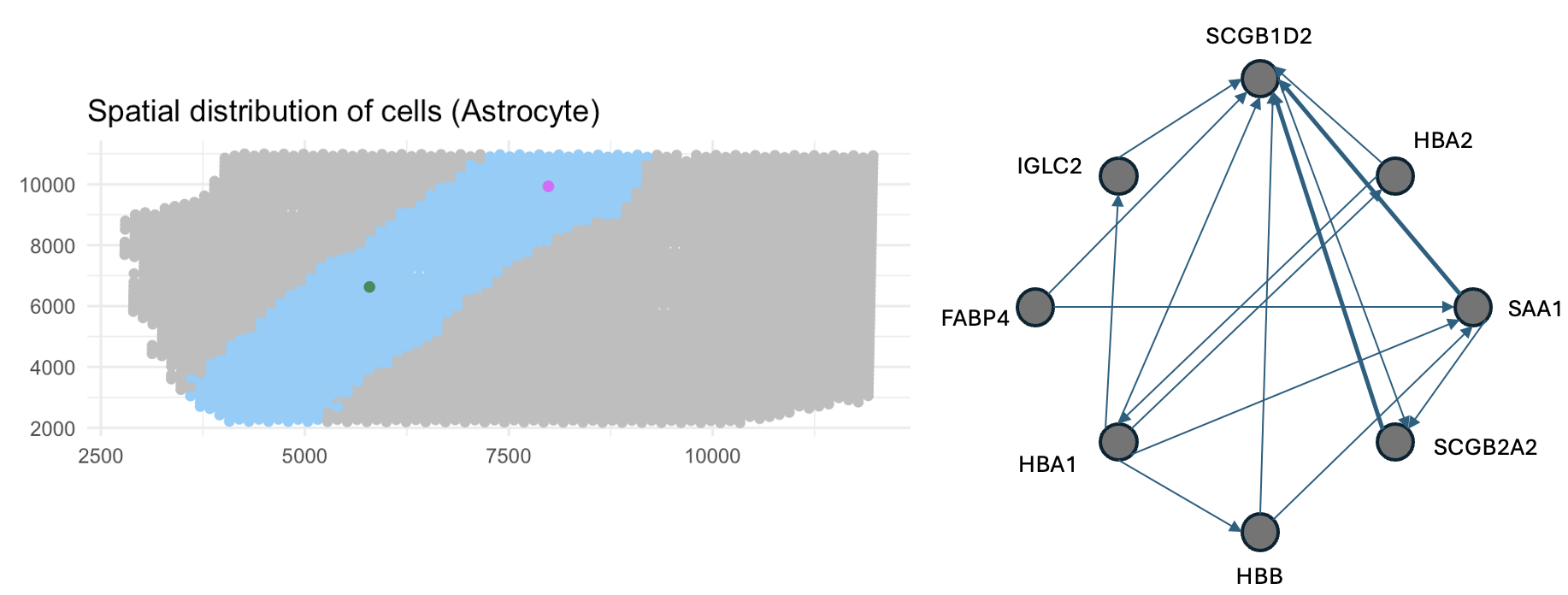}
     
	\caption{\textit{Left panel:} Spatial distribution of astrocytes (in light blue). Two individual spots are highlighted, for which we will show their individual graphs. \textit{Right panel:}  Union graph showing all possible edges at all spot locations. }
		\label{fig: union_type1}
\end{figure}

We depict the union of all the spot-level svGRNs as a cell-population-level summary for the excitatory neurons, other neurons, oligodendrocytes, and astrocytes in the right panels of Figures~\ref{fig: union_type3}, \ref{fig: union_type2}, \ref{fig: union_type4to7}, and \ref{fig: union_type1}, respectively. The edge width reflects the frequency of occurrence of the edges across all spot-level individual svGRNs. 
Highly connected hub genes across spot-level svGRNs tend to play crucial roles in gene
regulation as they are involved in many regulatory activities for many cells. 
In the excitatory neurons, \emph{SST}, a neuropeptide, plays a central role mostly in brain functions and cognitive behavior. A reduction in the expression of \emph{SST} could be an important indication of neurological disorders such as depression and Alzheimer's disease   \citep{song2021role}. \cite{alfimova2021relationship} suggests that the methylated \emph{YJEFN3} is involved with cognitive deficits such as schizophrenia spectrum disorders. 
 \emph{CXCL14} contributes towards a wide variety of biological processes such as neural development and immune process \citep{westrich2020multifarious, zhou2023role}.
In other neurons, \emph{BEX5} is a hub and is one of the members of the family of BEX proteins in humans, which is involved in signaling pathways and transcriptional regulations for neurodegenration and cell growth \citep{fernandez2015brain}. \emph{HMGB1} is responsible for a wide variety of mechanisms mostly related to aging and age-related diseases, inflammatory diseases, and immunity \citep{martinotti2015emerging, yuan2020high, guan2023comprehensive, ruggieri2024hmgb1}. 
    In Oligodendrocytes, 
    \emph{GFAP} is a biomarker for neurological disorders and could cause Alexander disease through mutation \citep{ quinlan2007gfap, yang2015glial}.

We select some spots, whose spatial locations are highlighted in the left panels of Figures \ref{fig: union_type3}-\ref{fig: union_type1}, for the visualization of spot-level svGRNs in Figures~\ref{fig: diffBand_type3}-\ref{fig: sameBand_Type1}. 
The edge width represents the posterior inclusion probability of the edge.
Many of them contain at least one feedback loop, which cannot be detected with existing DAG learning algorithms. For the neurons (Figures \ref{fig: diffBand_type3}-\ref{fig: sameBand_Type2}), we notice an interesting similarity in the structure of svGRNs within each band-like structure and the difference between the two bands. For example, the two svGRNs displayed in Figure \ref{fig: diffBand_type2}, one from each band, reveal notable differences in both edge inclusion and edge strengths. Conversely, the svGRNs corresponding to spots within the same band, e.g., the pair in Figure \ref{fig: otherNeuronBand1} and the pair in Figure \ref{fig: otherNeuronBand2} exhibit structural similarity with small variations in edge strengths. The larger dissimilarity between the bands indicates regulatory heterogeneity for cells of the same type, which may be due to microenvironmental factors.

\begin{figure}[h!]
    \centering
        \includegraphics[width = 12cm, height = 5cm]{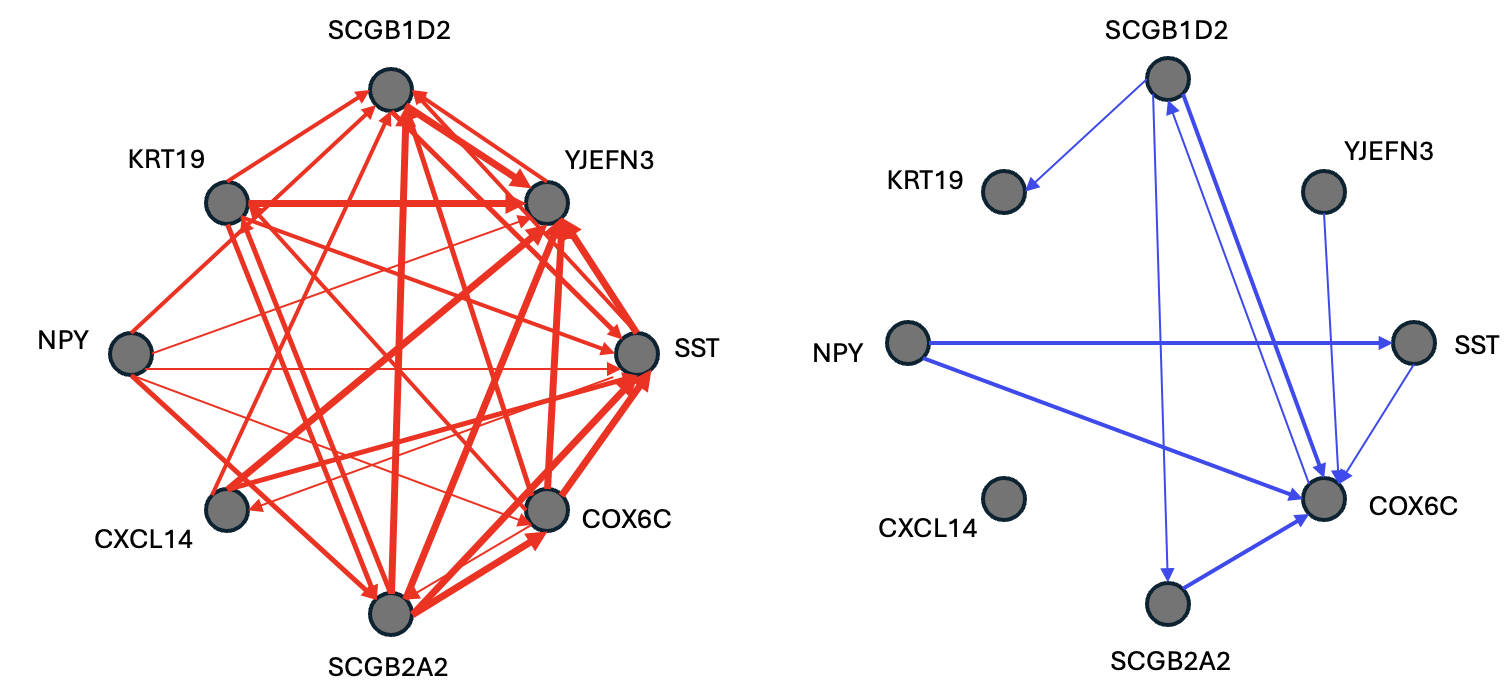}
	\caption{svGRNs for two different spots from two different bands for the excitatory neurons. The locations of the spots are depicted in Figure~\ref{fig: union_type3}.}
		\label{fig: diffBand_type3}
\end{figure}

\begin{figure}[h!]
  \centering
  
  \begin{subfigure}{0.8\textwidth}
    \centering
      \includegraphics[width = 12cm, height = 5cm]{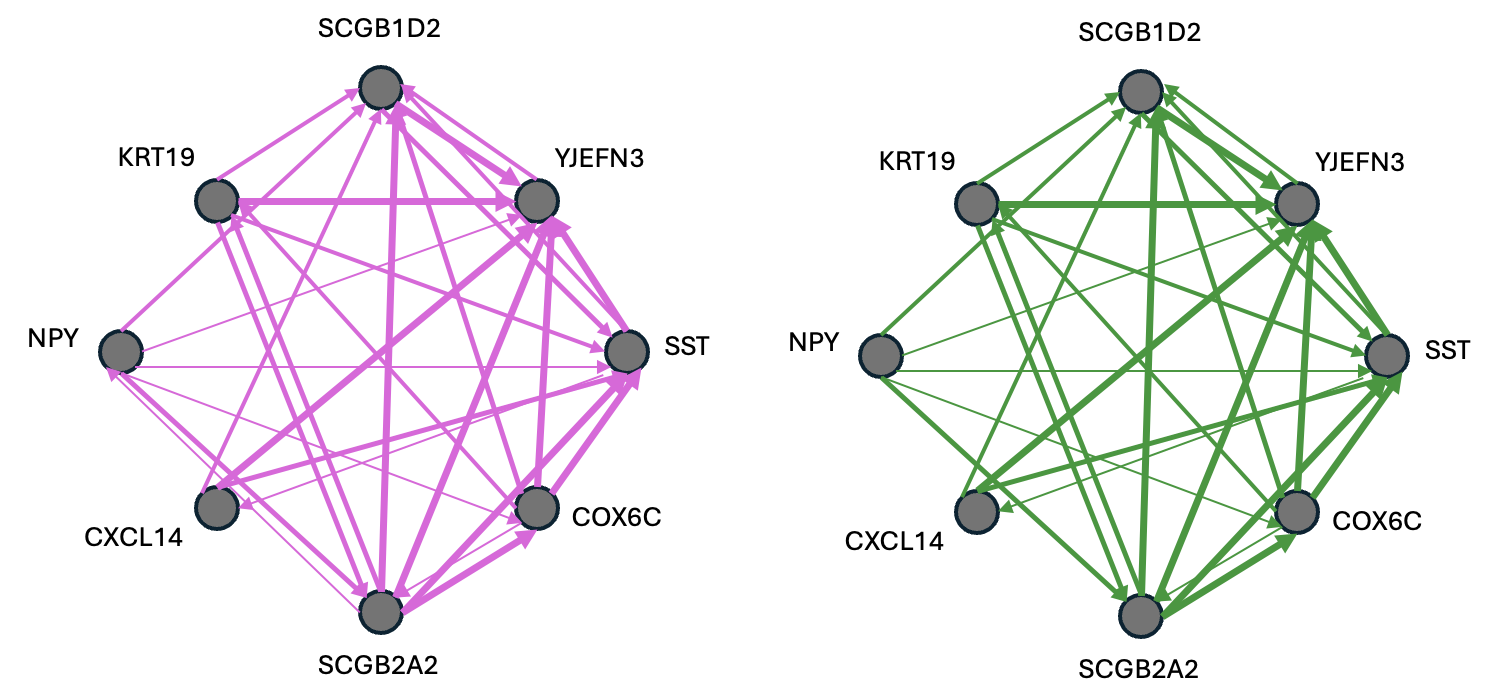}
    \caption{}
    \label{fig: excitatoryNeuronBand1}
  \end{subfigure}
  
  \begin{subfigure}{0.8\textwidth}
    \centering
     \includegraphics[width = 12cm, height = 5cm]{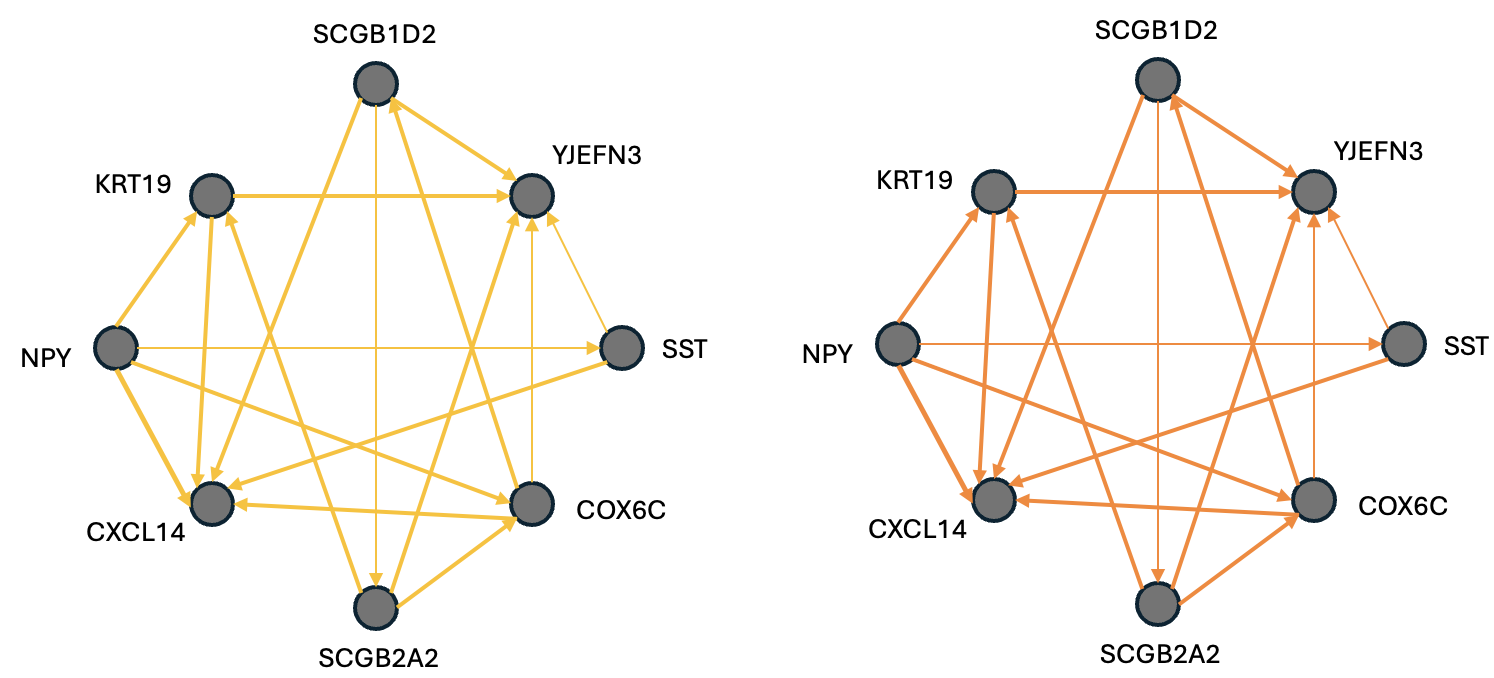}
    \caption{}
    \label{fig: excitatoryNeuronBand2}
  \end{subfigure}
  
  \caption{svGRNs for different spots from the excitatory neurons. The svGRN pairs from the upper and lower panel correspond to spot locations from the same band as plotted in Figure~\ref{fig: union_type3} (left panel).}
  \label{fig: sameBand_Type3}
\end{figure}

\begin{figure}[h!]
    \centering
        \includegraphics[width = 12cm, height = 5cm]{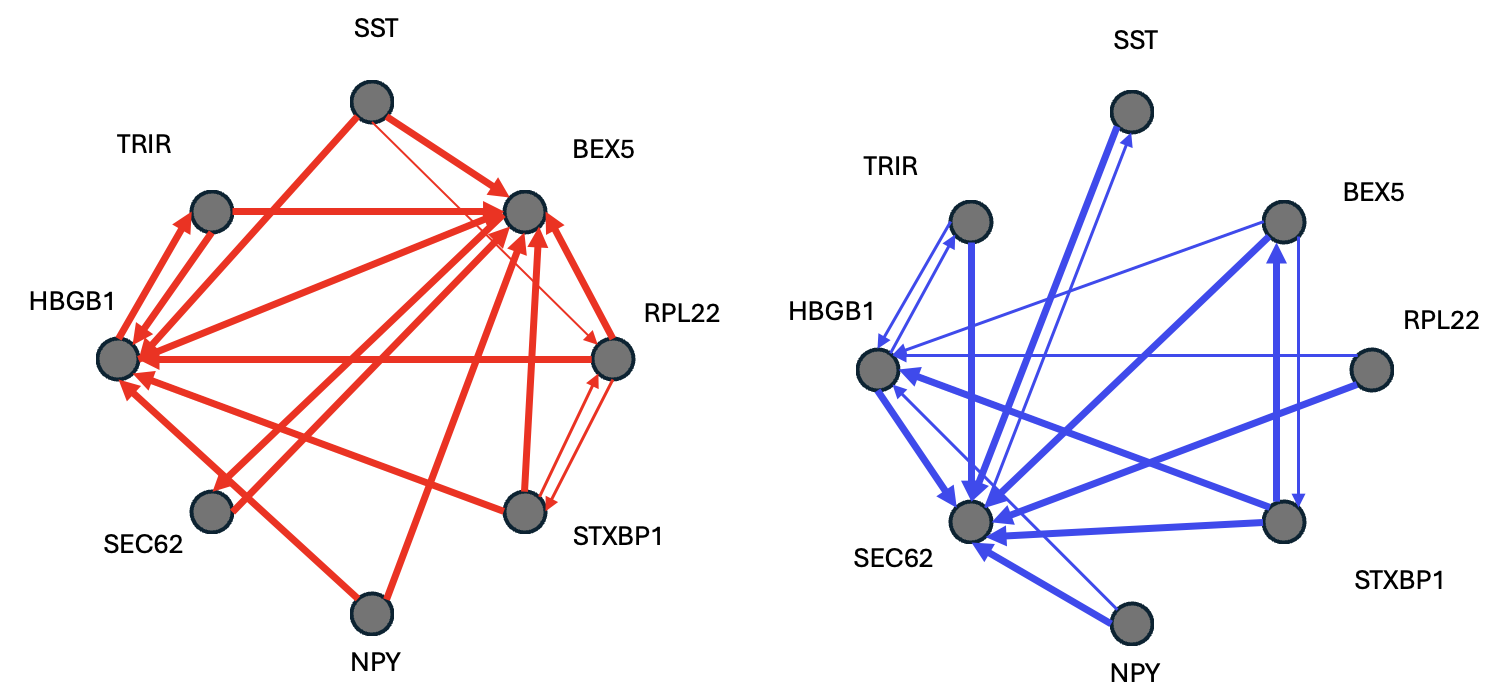}
	\caption{svGRNs for two different spots from two different bands for the other neurons. The locations of the spots are depicted in Figure~\ref{fig: union_type2}.}
		\label{fig: diffBand_type2}
\end{figure}

\begin{figure}[h!]
  \centering
  
  \begin{subfigure}{0.8\textwidth}
    \centering
      \includegraphics[width = 12cm, height = 5cm]{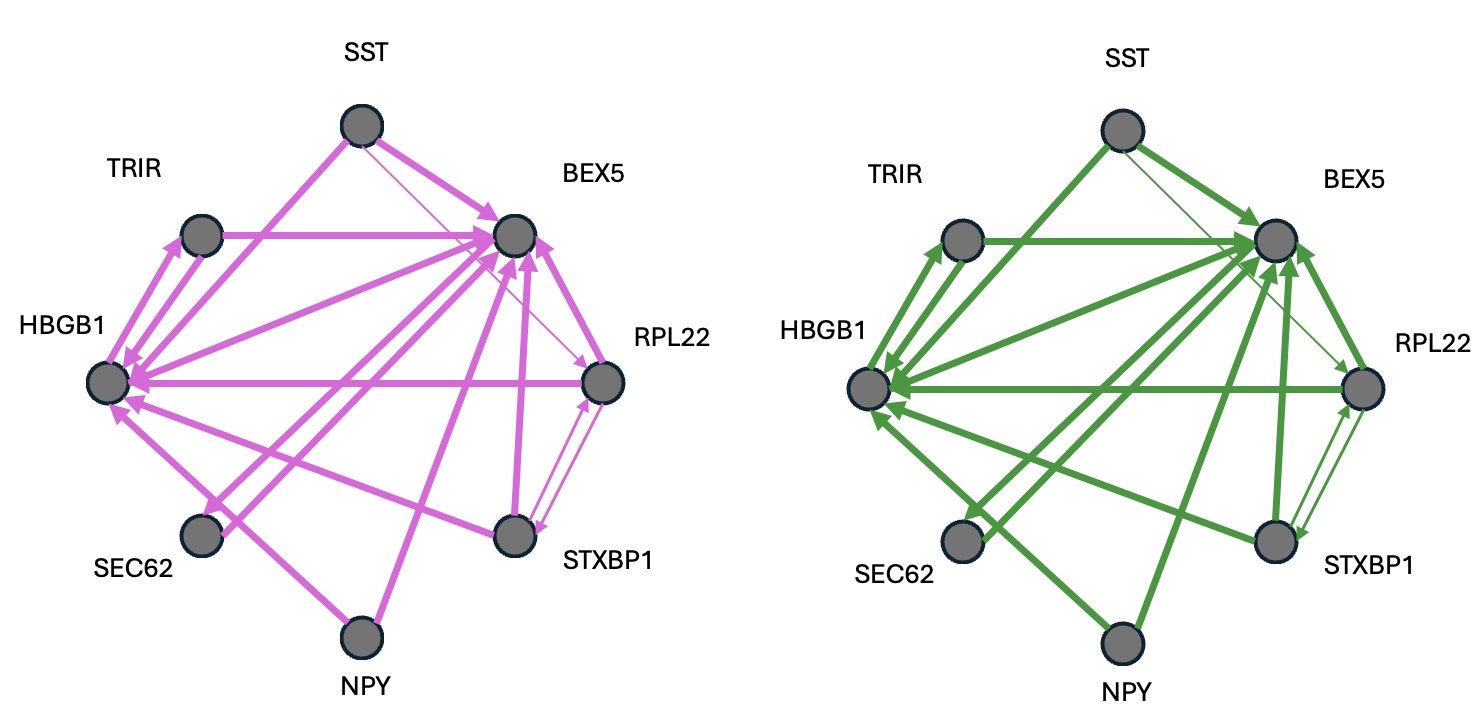}
      \caption{}
    \label{fig: otherNeuronBand1}
  \end{subfigure}
  
  \begin{subfigure}{0.8\textwidth}
    \centering
     \includegraphics[width = 12cm, height = 5cm]{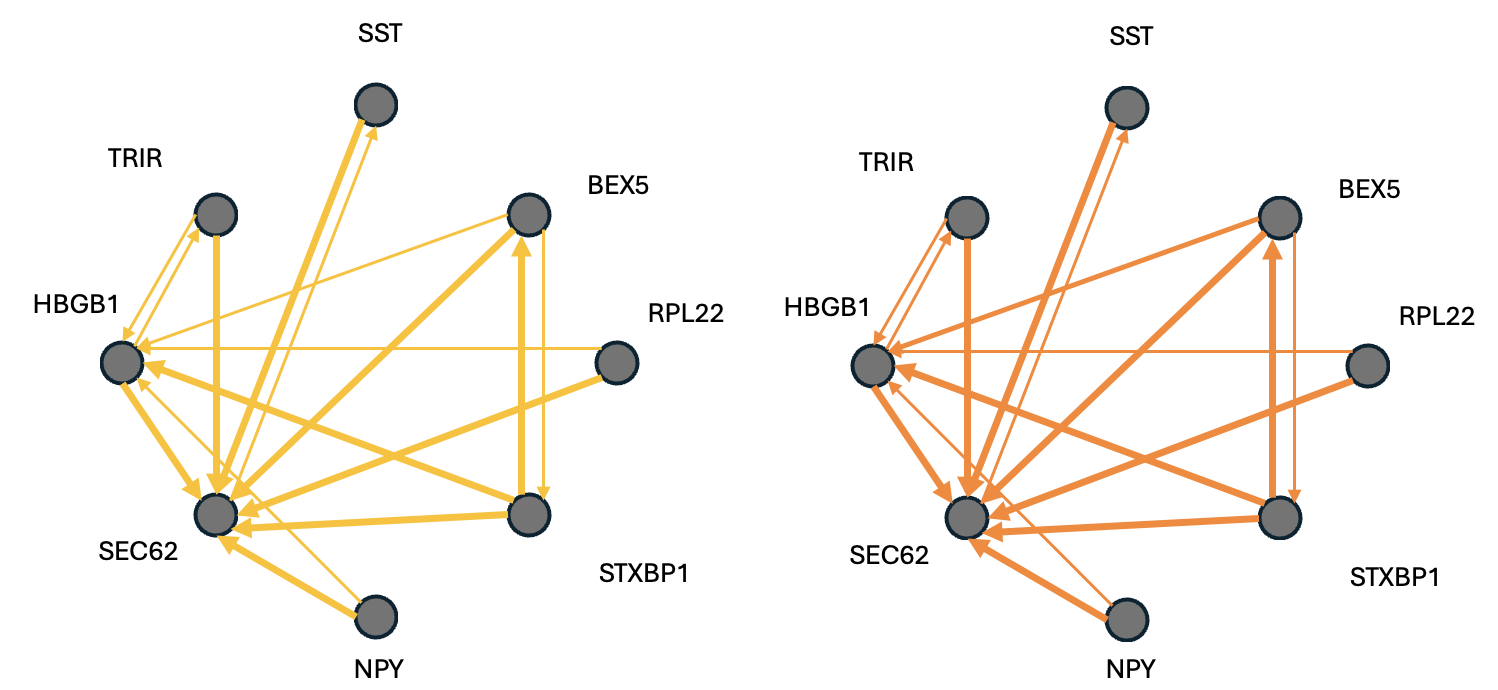}
    \caption{}
    \label{fig: otherNeuronBand2}
  \end{subfigure}
  
  \caption{svGRNs for different spots from the other neurons. The svGRN pairs from the upper and lower panel correspond to spot locations fromthe  same band as plotted in Figure~\ref{fig: union_type2} (left panel).}
  \label{fig: sameBand_Type2}
\end{figure}

\begin{figure}[h!]
  \centering
  
  \begin{subfigure}{0.8\textwidth}
    \centering
      \includegraphics[width = 12cm, height = 5cm]{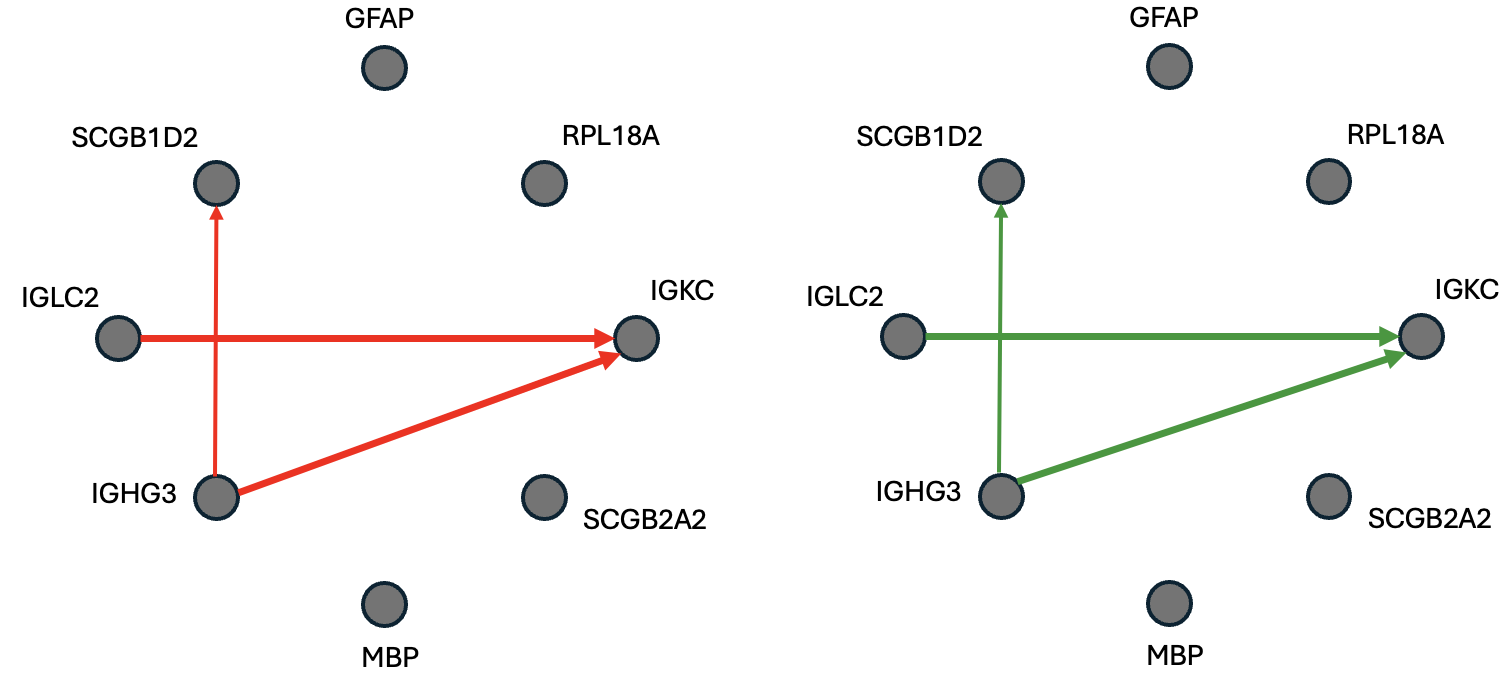}
    \label{fig: OligoBand12}
  \end{subfigure}
  
  \begin{subfigure}{0.8\textwidth}
    \centering
     \includegraphics[width = 12cm, height = 5cm]{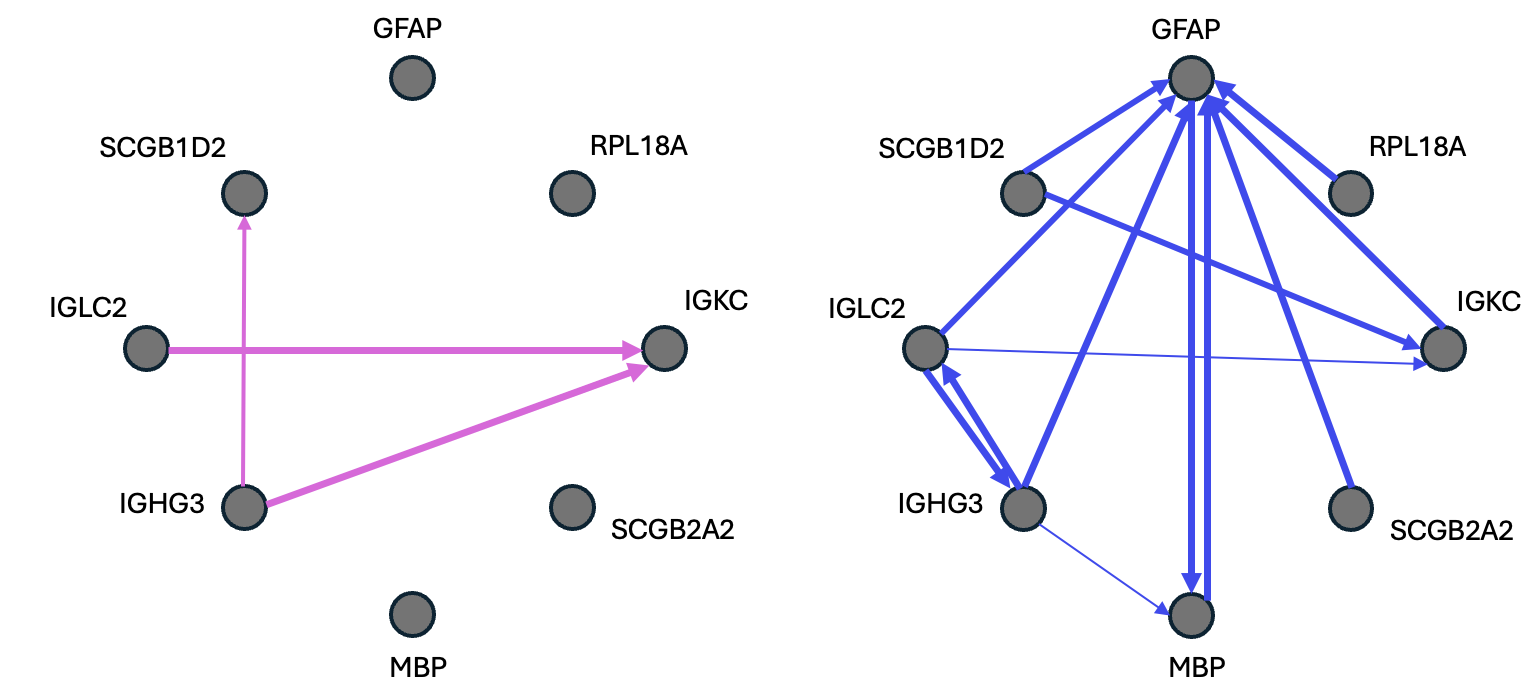}
    \label{fig: OligoBand34}
  \end{subfigure}
  
  \caption{svGRNs for different spots from the oligodendrocytes. The svGRNs are selected corresponding to the spots from the consecutive bands as shown in  Figure~\ref{fig: union_type4to7} (left panel).}
  \label{fig: diffBand_Type4to7}
\end{figure}

\begin{figure}[h!]
    \centering
        \includegraphics[width = 12cm, height = 5.5cm]{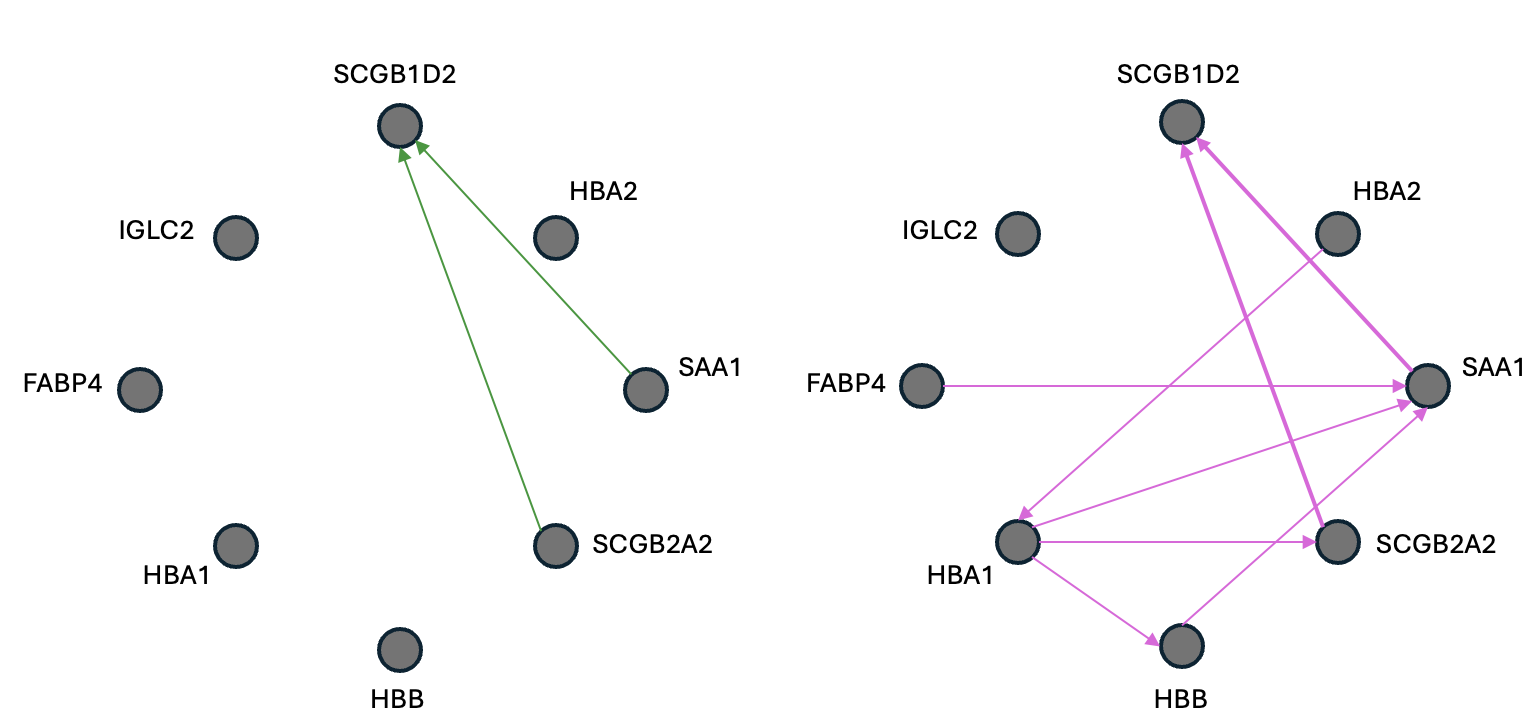}
	\caption{svGRNs for different spots from the astrocytes. The svGRNs are selected corresponding to the two spots as shown in  Figure~\ref{fig: union_type1} (left panel).}
		\label{fig: sameBand_Type1}
\end{figure}

Such spatial regulatory heterogeneity cannot be captured by conventional spatial transcriptomic clustering algorithms. For example, we applied the Dimension\mbox{-}Reduction Spatial\mbox{-}Clustering (DR-SC) \citep{liu2022joint} designed for spatial transcriptomics data as well as the traditional k-means algorithm with $k = 2$. 
The left panel of Figure \ref{fig: SpaCl} presents the k-means clustering based solely on gene expression data, while the middle panel shows the clustering obtained using a combined and scaled version of both gene expression and spatial information. Clearly, k-means failed to detect the two bands as two separate clusters. In the former case, the two clusters are mixed within each band with no clear separation. 
In the latter case, while some localization of the two clusters is observed, the clusters do not correspond to the two spatial bands. Instead, each band contains roughly two clusters with some interspersed points in the middle of the bands. The right panel of Figure~\ref{fig: SpaCl} shows that only one cluster was found by DR-SC. Failure to detect the distinct two bands could be explained by the unimodal marginal distributions of the gene expression (results not shown).
The proposed method could successfully detect the two spatial bands, which could be novel neuron subtypes, by leveraging the underlying differences in their GRNs in addition to the gene expression differences.
\begin{figure}[h!]
    \centering
        \includegraphics[width = \textwidth]{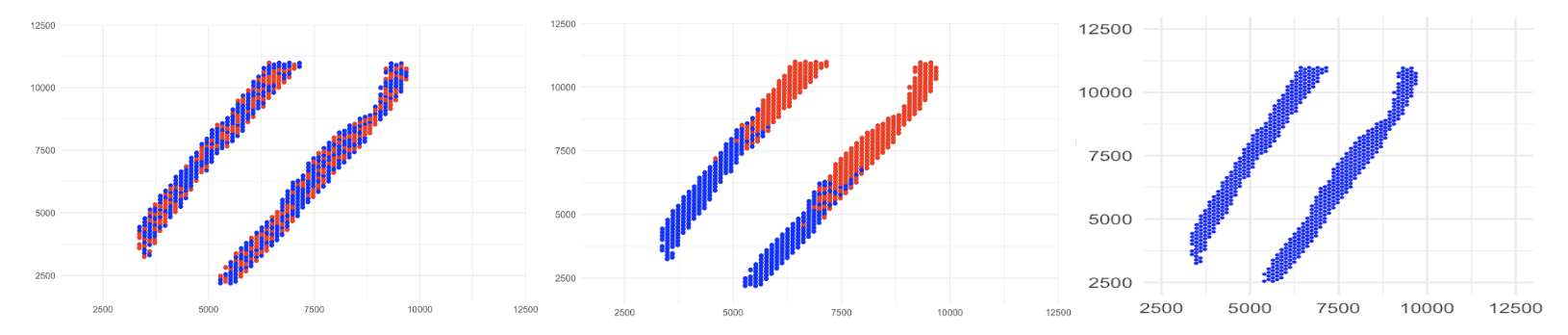}
	\caption{\textit{Left panel:} k-means clustering applied to gene expression data. \textit{Middle panel:} k-means applied to combined and scaled gene expression and spatial information. \textit{Right panel:} DR-SC clustering applied to the combined gene expression and spatial location data.}
		\label{fig: SpaCl}
\end{figure}
The identification of cell subtypes is biologically plausible. As stated by \cite{maynard2021transcriptome}, the human dorsolateral prefrontal cortex can be characterized by its laminar organization, that is, by the arrangement of cortical neurons into six distinct layers. The cell type and spots within different cortical layers can exhibit distinct gene expression profiles and varying connectivity patterns of genes. Moreover, advancements in single-cell transcriptomics and the continuous refinement of cell type classifications have uncovered the existence of remarkable diversity in cortical cell subtypes. Refer to \cite{miller2019shared} for details on cellular diversity and subtypes in the human cerebral cortex. This diversity arises during development where specialized groups of immature cells, known as progenitor pools, give rise to a variety of cell types and subtypes. Consequently, within a single cell type, there may exist distinct cellular microenvironments or multiple cell subtypes, potentially leading to differences in gene regulatory relationships.

In summary, the proposed method detected svGRNs with feedback loops based on spatial transcriptomic data using spatial locations as covariates. Moreover, the inferred graph structures reveal distinct clusters corresponding to separated spatial bands, allowing the identification of plausible cell subtypes driven by distinct GRNs. While GRNs within the same band may exhibit structural similarities, the variation in edge strength highlights their spatial dynamics. This suggests the potential influence of the cellular microenvironment on GRNs, emphasizing the role of spatially varying connectivity within the tissue context.

\section{Discussion}\label{Section 6}

We have proposed a novel Bayesian nonparametric directed cyclic graph model, BNP-DCGx, which leverages covariate information in heterogeneous data through a covariate-dependent random partition model. The key feature is the use of an intermediary partition layer between covariates and graph structures, which enables us to satisfy the stability condition, an eigenvalue constraint on the SEM coefficient matrix, while allowing graph structures and edge strengths to vary continuously with covariates. We developed a parallel-tempered MCMC algorithm and ensured smoothly varying graphs through partition averaging from a set of discrete graphs originated from the random partition. The efficacy of the proposed method was demonstrated through simulations and an application to spatial transcriptomics data from human dorsolateral prefrontal cortex. In the real data analysis, the method identified graph structures with feedback loops that vary across spatial locations within a single cell type. Notably, the inferred regulatory networks revealed distinct clusters corresponding to separated spatial bands, identifying plausible cell subtypes driven by distinct gene regulatory mechanisms. These regulatory subtypes were not detected by conventional clustering methods based on gene expression alone, highlighting the value of network-level analysis for uncovering hidden heterogeneity. The variation in edge strengths within spatial bands further suggests the influence of local microenvironmental factors on gene regulation, emphasizing the role of spatially varying connectivity within tissue context.

Future research directions include extending the model to accommodate confounders, exploring non-linear relationships beyond the current linear structural equation framework, and investigating alternative approaches for constructing covariate-dependent cyclic graphs. The computational efficiency could be improved through variational Bayes approximations for large graphs or datasets. The identification of regulatory subtypes based on network differences opens new avenues for precision medicine, where treatments could be tailored to an individual's regulatory profile rather than expression profile alone. As spatial multi-omics technologies continue to advance, methods that integrate spatial information with regulatory network inference will become increasingly valuable for understanding the complex organization of biological systems.
Beyond spatial transcriptomics, the proposed framework has broad applicability in diverse fields where feedback mechanisms and heterogeneity are expected, such as dynamic brain networks in neuroscience, metabolic and signaling pathways in systems biology, and simultaneous equation models in econometrics.

\bibliographystyle{apalike}
\bibliography{biblio}

@article{swendsen1986replica,
  title={Replica {M}onte {C}arlo simulation of spin-glasses},
  author={Swendsen, Robert H and Wang, Jian-Sheng},
  journal={Physical review letters},
  volume={57},
  number={21},
  pages={2607},
  year={1986},
  publisher={APS}
}

@inproceedings{geyer1991markov,
  title={Markov chain {M}onte {C}arlo maximum likelihood},
  author={Geyer, Charles J},
  booktitle={Computing science and statistics: Proceedings of the 23rd Symposium on the Interface},
  volume={156163},
  year={1991},
  organization={New York}
}

@inproceedings{peal1985bayesian,
  title={Bayesian networks: A model of self-activated memory for evidential reasoning},
  author={Peal, Judea},
  booktitle={Proceedings of the Annual Meeting of the Cognitive Science Society},
  volume={7},
  year={1985}
}

@article{mason1953feedback,
  title={Feedback theory-some properties of signal flow graphs},
  author={Mason, Samuel J},
  journal={Proceedings of the IRE},
  volume={41},
  number={9},
  pages={1144--1156},
  year={1953},
  publisher={IEEE}
}

@article{haavelmo1943statistical,
  title={The statistical implications of a system of simultaneous equations},
  author={Haavelmo, Trygve},
  journal={Econometrica},
  pages={1--12},
  year={1943},
  publisher={JSTOR}
}

@inproceedings{spirtes1995directed,
author = {Spirtes, Peter},
title = {Directed cyclic graphical representations of feedback models},
year = {1995},
isbn = {1558603859},
publisher = {Morgan Kaufmann Publishers Inc.},
address = {San Francisco, CA, USA},
booktitle = {Proceedings of the Eleventh Conference on Uncertainty in Artificial Intelligence},
pages = {491–498},
numpages = {8},
location = {Montr\'{e}al, Qu\'{e}, Canada},
series = {UAI'95}
}

@article{ni2018reciprocal,
  title={Reciprocal graphical models for integrative gene regulatory network analysis},
  author={Ni, Yang and Ji, Yuan and M{\"u}ller, Peter},
  journal={Bayesian Analysis},
  volume={13},
  number={4},
  pages={1095--1110},
  year={2018}
}

@article{mooij2011causal,
  title={On causal discovery with cyclic additive noise models},
  author={Mooij, Joris M and Janzing, Dominik and Heskes, Tom and Sch{\"o}lkopf, Bernhard},
  journal={Advances in Neural Information Processing Systems},
  volume={24},
  year={2011}
}

@article{zhou2023individualized,
  title={Individualized causal discovery with latent trajectory embedded {B}ayesian networks},
  author={Zhou, Fangting and He, Kejun and Ni, Yang},
  journal={Biometrics},
  volume={79},
  number={4},
  pages={3191--3202},
  year={2023},
  publisher={Wiley Online Library}
}

@article{yao2025robust,
  title={Robust {B}ayesian graphical regression models for assessing tumor heterogeneity in proteomic networks},
  author={Yao, Tsung-Hung and Ni, Yang and Bhadra, Anindya and Kang, Jian and Baladandayuthapani, Veerabhadran},
  journal={Biometrics},
  volume={81},
  number={1},
  pages={ujae160},
  year={2025},
  publisher={Oxford University Press}
}

@article{chen2023probabilistic,
  title={Probabilistic Graphical Modeling under Heterogeneity},
  author={Chen, Liying and Acharyya, Satwik and Luo, Chunyu and Ni, Yang and Baladandayuthapani, Veerabhadran},
  journal={bioRxiv},
  pages={2023--10},
  year={2023},
  publisher={Cold Spring Harbor Laboratory}
}

@article{ni2022bayesian,
  title={Bayesian covariate-dependent {G}aussian graphical models with varying structure},
  author={Ni, Yang and Stingo, Francesco C and Baladandayuthapani, Veerabhadran},
  journal={Journal of Machine Learning Research},
  volume={23},
  number={242},
  pages={1--29},
  year={2022}
}

@article{bollen1989structural,
  title={Structural Equations with Latent Variables.},
  author={Bollen, KA},
  year={1989},
  publisher={New York: Wiley}
}

@article{lacerda2012discovering,
  title={Discovering cyclic causal models by independent components analysis},
  author={Lacerda, Gustavo and Spirtes, Peter L and Ramsey, Joseph and Hoyer, Patrik O},
  journal={arXiv preprint arXiv:1206.3273},
  year={2012}
}

@article{niu2023covariate,
  title={Covariate-assisted Bayesian graph learning for heterogeneous data},
  author={Niu, Yabo and Ni, Yang and Pati, Debdeep and Mallick, Bani K},
  journal={Journal of the American Statistical Association},
  pages={1--15},
  year={2023},
  publisher={Taylor \& Francis}
}

@article{muller2011product,
  title={A product partition model with regression on covariates},
  author={M{\"u}ller, Peter and Quintana, Fernando and Rosner, Gary L},
  journal={Journal of Computational and Graphical Statistics},
  volume={20},
  number={1},
  pages={260--278},
  year={2011},
  publisher={Taylor \& Francis}
}

@article{kuleshov2016enrichr,
  title={Enrichr: a comprehensive gene set enrichment analysis web server 2016 update},
  author={Kuleshov, Maxim V and Jones, Matthew R and Rouillard, Andrew D and Fernandez, Nicolas F and Duan, Qiaonan and Wang, Zichen and Koplev, Simon and Jenkins, Sherry L and Jagodnik, Kathleen M and Lachmann, Alexander and others},
  journal={Nucleic acids research},
  volume={44},
  number={W1},
  pages={W90--W97},
  year={2016},
  publisher={Oxford University Press}
}

@article{neal2000markov,
  title={Markov chain sampling methods for Dirichlet process mixture models},
  author={Neal, Radford M},
  journal={Journal of computational and graphical statistics},
  volume={9},
  number={2},
  pages={249--265},
  year={2000},
  publisher={Taylor \& Francis}
}

@article{xie2020identifying,
  title={Identifying disease-associated biomarker network features through conditional graphical model},
  author={Xie, Shanghong and Li, Xiang and McColgan, Peter and Scahill, Rachael I and Zeng, Donglin and Wang, Yuanjia},
  journal={Biometrics},
  volume={76},
  number={3},
  pages={995--1006},
  year={2020},
  publisher={Oxford University Press}
}

@article{heppner1983tumor,
  title={Tumor heterogeneity: biological implications and therapeutic consequences},
  author={Heppner, Gloria H and Miller, Bonnie E},
  journal={Cancer and Metastasis Reviews},
  volume={2},
  pages={5--23},
  year={1983},
  publisher={Springer}
}

@article{yin2011sparse,
  title={A sparse conditional Gaussian graphical model for analysis of genetical genomics data},
  author={Yin, Jianxin and Li, Hongzhe},
  journal={The annals of applied statistics},
  volume={5},
  number={4},
  pages={2630},
  year={2011},
  publisher={NIH Public Access}
}

@article{li2012sparse,
  title={Sparse estimation of conditional graphical models with application to gene networks},
  author={Li, Bing and Chun, Hyonho and Zhao, Hongyu},
  journal={Journal of the American Statistical Association},
  volume={107},
  number={497},
  pages={152--167},
  year={2012},
  publisher={Taylor \& Francis}
}

@article{cai2013covariate,
  title={Covariate-adjusted precision matrix estimation with an application in genetical genomics},
  author={Cai, T Tony and Li, Hongzhe and Liu, Weidong and Xie, Jichun},
  journal={Biometrika},
  volume={100},
  number={1},
  pages={139--156},
  year={2013},
  publisher={Oxford University Press}
}

@article{JMLR:v17:16-004,
  author  = {Jiahe Lin and Sumanta Basu and Moulinath Banerjee and George Michailidis},
  title   = {Penalized Maximum Likelihood Estimation of Multi-layered Gaussian Graphical Models},
  journal = {Journal of Machine Learning Research},
  year    = {2016},
  volume  = {17},
  number  = {146},
  pages   = {1--51},
  url     = {http://jmlr.org/papers/v17/16-004.html}
}

@article{chen2016asymptotically,
  title={Asymptotically normal and efficient estimation of covariate-adjusted Gaussian graphical model},
  author={Chen, Mengjie and Ren, Zhao and Zhao, Hongyu and Zhou, Harrison},
  journal={Journal of the American Statistical Association},
  volume={111},
  number={513},
  pages={394--406},
  year={2016},
  publisher={Taylor \& Francis}
}

@article{bhadra2013joint,
  title={Joint high-dimensional Bayesian variable and covariance selection with an application to eQTL analysis},
  author={Bhadra, Anindya and Mallick, Bani K},
  journal={Biometrics},
  volume={69},
  number={2},
  pages={447--457},
  year={2013},
  publisher={Oxford University Press}
}

@article{deshpande2019simultaneous,
  title={Simultaneous variable and covariance selection with the multivariate spike-and-slab lasso},
  author={Deshpande, Sameer K and Ro{\v{c}}kov{\'a}, Veronika and George, Edward I},
  journal={Journal of Computational and Graphical Statistics},
  volume={28},
  number={4},
  pages={921--931},
  year={2019},
  publisher={Taylor \& Francis}
}

@article{niu2020bayesian,
  title={Bayesian variable selection in multivariate nonlinear regression with graph structures},
  author={Niu, Yabo and Guha, Nilabja and De, Debkumar and Bhadra, Anindya and Baladandayuthapani, Veerabhadran and Mallick, Bani K},
  journal={arXiv preprint arXiv:2010.14638},
  year={2020}
}

@article{fox2015bayesian,
  title={Bayesian nonparametric covariance regression},
  author={Fox, Emily B and Dunson, David B},
  journal={The Journal of Machine Learning Research},
  volume={16},
  number={1},
  pages={2501--2542},
  year={2015},
  publisher={JMLR. org}
}

@article{zeng2024bayesian,
  title={Bayesian Covariate-Dependent Graph Learning with a Dual Group Spike-and-Slab Prior},
  author={Zeng, Zijian and Li, Meng and Vannucci, Marina},
  journal={arXiv preprint arXiv:2409.17404},
  year={2024}
}

@article{zhou2010time,
  title={Time varying undirected graphs},
  author={Zhou, Shuheng and Lafferty, John and Wasserman, Larry},
  journal={Machine Learning},
  volume={80},
  pages={295--319},
  year={2010},
  publisher={Springer}
}

@article{lee2018nonparametric,
  title={Nonparametric finite mixture of Gaussian graphical models},
  author={Lee, Kevin H and Xue, Lingzhou},
  journal={Technometrics},
  volume={60},
  number={4},
  pages={511--521},
  year={2018},
  publisher={Taylor \& Francis}
}

@article{guo2011joint,
  title={Joint estimation of multiple graphical models},
  author={Guo, Jian and Levina, Elizaveta and Michailidis, George and Zhu, Ji},
  journal={Biometrika},
  volume={98},
  number={1},
  pages={1--15},
  year={2011},
  publisher={Oxford University Press}
}

@article{danaher2014joint,
  title={The joint graphical lasso for inverse covariance estimation across multiple classes},
  author={Danaher, Patrick and Wang, Pei and Witten, Daniela M},
  journal={Journal of the Royal Statistical Society Series B: Statistical Methodology},
  volume={76},
  number={2},
  pages={373--397},
  year={2014},
  publisher={Oxford University Press}
}

@article{oates2014joint,
  title={Joint estimation of multiple related biological networks},
  author={Oates, Chris J and Korkola, Jim and Gray, Joe W and Mukherjee, Sach},
  journal={The Annals of Applied Statistics},
  volume={8},
  number={3},
  pages={1892--1919},
  year={2014},
  publisher={Institute of Mathematical Statistics}
}

@article{peterson2015bayesian,
  title={Bayesian inference of multiple Gaussian graphical models},
  author={Peterson, Christine and Stingo, Francesco C and Vannucci, Marina},
  journal={Journal of the American Statistical Association},
  volume={110},
  number={509},
  pages={159--174},
  year={2015},
  publisher={Taylor \& Francis}
}

@article{yajima2015detecting,
  title={Detecting differential patterns of interaction in molecular pathways},
  author={Yajima, Masanao and Telesca, Donatello and Ji, Yuan and M{\"u}ller, Peter},
  journal={Biostatistics},
  volume={16},
  number={2},
  pages={240--251},
  year={2015},
  publisher={Oxford University Press}
}

@article{xie2016joint,
  title={Joint estimation of multiple dependent Gaussian graphical models with applications to mouse genomics},
  author={Xie, Yuying and Liu, Yufeng and Valdar, William},
  journal={Biometrika},
  volume={103},
  number={3},
  pages={493--511},
  year={2016},
  publisher={Oxford University Press}
}

@article{ni2018heterogeneous,
  title={Heterogeneous reciprocal graphical models},
  author={Ni, Yang and M{\"u}ller, Peter and Zhu, Yitan and Ji, Yuan},
  journal={Biometrics},
  volume={74},
  number={2},
  pages={606--615},
  year={2018},
  publisher={Oxford University Press}
}

@article{shaddox2018bayesian,
  title={A Bayesian approach for learning gene networks underlying disease severity in COPD},
  author={Shaddox, Elin and Stingo, Francesco C and Peterson, Christine B and Jacobson, Sean and Cruickshank-Quinn, Charmion and Kechris, Katerina and Bowler, Russell and Vannucci, Marina},
  journal={Statistics in Biosciences},
  volume={10},
  pages={59--85},
  year={2018},
  publisher={Springer}
}

@article{mitra2016bayesian,
  title={Bayesian graphical models for differential pathways},
  author={Mitra, Riten and M{\"u}ller, Peter and Ji, Yuan},
  year={2016}
}

@article{tan2017bayesian,
  title={Bayesian inference for multiple Gaussian graphical models with application to metabolic association networks},
  author={Tan, Linda SL and Jasra, Ajay and De Iorio, Maria and Ebbels, Timothy MD},
  year={2017}
}

@article{lin2017joint,
  title={On joint estimation of Gaussian graphical models for spatial and temporal data},
  author={Lin, Zhixiang and Wang, Tao and Yang, Can and Zhao, Hongyu},
  journal={Biometrics},
  volume={73},
  number={3},
  pages={769--779},
  year={2017},
  publisher={Wiley Online Library}
}

@article{talluri2014bayesian,
  title={Bayesian sparse graphical models and their mixtures},
  author={Talluri, Rajesh and Baladandayuthapani, Veerabhadran and Mallick, Bani K},
  journal={Stat},
  volume={3},
  number={1},
  pages={109--125},
  year={2014},
  publisher={Wiley Online Library}
}

@article{rodriguez2011sparse,
  title={Sparse covariance estimation in heterogeneous samples},
  author={Rodriguez, Abel and Lenkoski, Alex and Dobra, Adrian},
  journal={Electronic journal of statistics},
  volume={5},
  pages={981},
  year={2011},
  publisher={NIH Public Access}
}

@article{ni2019bayesian,
  title={Bayesian graphical regression},
  author={Ni, Yang and Stingo, Francesco C and Baladandayuthapani, Veerabhadran},
  journal={Journal of the American Statistical Association},
  volume={114},
  number={525},
  pages={184--197},
  year={2019},
  publisher={Taylor \& Francis}
}

@article{zhang2023high,
  title={High-dimensional Gaussian graphical regression models with covariates},
  author={Zhang, Jingfei and Li, Yi},
  journal={Journal of the American Statistical Association},
  volume={118},
  number={543},
  pages={2088--2100},
  year={2023},
  publisher={Taylor \& Francis}
}

@article{liu2010graph,
  title={Graph-valued regression},
  author={Liu, Han and Chen, Xi and Wasserman, Larry and Lafferty, John},
  journal={Advances in Neural Information Processing Systems},
  volume={23},
  year={2010}
}

@article{sagar2022bayesian,
  title={Bayesian Covariate-Dependent Quantile Directed Acyclic Graphical Models for Individualized Inference},
  author={Sagar, Ksheera and Ni, Yang and Baladandayuthapani, Veerabhadran and Bhadra, Anindya},
  journal={arXiv preprint arXiv:2210.08096},
  year={2022}
}

@article{wang2022bayesian,
  title={Bayesian edge regression in undirected graphical models to characterize interpatient heterogeneity in cancer},
  author={Wang, Zeya and Baladandayuthapani, Veerabhadran and Kaseb, Ahmed O and Amin, Hesham M and Hassan, Manal M and Wang, Wenyi and Morris, Jeffrey S},
  journal={Journal of the American Statistical Association},
  volume={117},
  number={538},
  pages={533--546},
  year={2022},
  publisher={Taylor \& Francis}
}

@article{maynard2021transcriptome,
  title={Transcriptome-scale spatial gene expression in the human dorsolateral prefrontal cortex},
  author={Maynard, Kristen R and Collado-Torres, Leonardo and Weber, Lukas M and Uytingco, Cedric and Barry, Brianna K and Williams, Stephen R and Catallini, Joseph L and Tran, Matthew N and Besich, Zachary and Tippani, Madhavi and others},
  journal={Nature neuroscience},
  volume={24},
  number={3},
  pages={425--436},
  year={2021},
  publisher={Nature Publishing Group US New York}
}

@inproceedings{richardson1996discovery,
author = {Richardson, Thomas},
title = {A discovery algorithm for directed cyclic graphs},
year = {1996},
isbn = {155860412X},
publisher = {Morgan Kaufmann Publishers Inc.},
address = {San Francisco, CA, USA},
booktitle = {Proceedings of the Twelfth International Conference on Uncertainty in Artificial Intelligence},
pages = {454–461},
numpages = {8},
location = {Portland, OR},
series = {UAI'96}
}

@article{chen2023spatial,
  title={Spatial transcriptomic technologies},
  author={Chen, Tsai-Ying and You, Li and Hardillo, Jose Angelito U and Chien, Miao-Ping},
  journal={Cells},
  volume={12},
  number={16},
  pages={2042},
  year={2023},
  publisher={MDPI}
}

@article{miller2019shared,
  title={Shared and derived features of cellular diversity in the human cerebral cortex},
  author={Miller, Daniel J and Bhaduri, Aparna and Sestan, Nenad and Kriegstein, Arnold},
  journal={Current opinion in neurobiology},
  volume={56},
  pages={117--124},
  year={2019},
  publisher={Elsevier}
}

@article{fisher1970correspondence,
  title={A correspondence principle for simultaneous equation models},
  author={Fisher, Franklin M},
  journal={Econometrica: Journal of the Econometric Society},
  pages={73--92},
  year={1970},
  publisher={JSTOR}
}

@Article{hao2023,
    author = {Yuhan Hao and Tim Stuart and Madeline H Kowalski and
      Saket Choudhary and Paul Hoffman and Austin Hartman and Avi
      Srivastava and Gesmira Molla and Shaista Madad and Carlos
      Fernandez-Granda and Rahul Satija},
    title = {Dictionary learning for integrative, multimodal and
      scalable single-cell analysis},
    journal = {Nature Biotechnology},
    year = {2023},
    doi = {10.1038/s41587-023-01767-y},
    url = {https://doi.org/10.1038/s41587-023-01767-y},
  }

@Article{hao2021,
    author = {Yuhan Hao and Stephanie Hao and Erica Andersen-Nissen and
      William M. Mauck III and Shiwei Zheng and Andrew Butler and
      Maddie J. Lee and Aaron J. Wilk and Charlotte Darby and Michael
      Zagar and Paul Hoffman and Marlon Stoeckius and Efthymia Papalexi
      and Eleni P. Mimitou and Jaison Jain and Avi Srivastava and Tim
      Stuart and Lamar B. Fleming and Bertrand Yeung and Angela J.
      Rogers and Juliana M. McElrath and Catherine A. Blish and Raphael
      Gottardo and Peter Smibert and Rahul Satija},
    title = {Integrated analysis of multimodal single-cell data},
    journal = {Cell},
    year = {2021},
    doi = {10.1016/j.cell.2021.04.048},
    url = {https://doi.org/10.1016/j.cell.2021.04.048},
  }

@Article{stuart2019,
    author = {Tim Stuart and Andrew Butler and Paul Hoffman and
      Christoph Hafemeister and Efthymia Papalexi and William M Mauck
      III and Yuhan Hao and Marlon Stoeckius and Peter Smibert and
      Rahul Satija},
    title = {Comprehensive Integration of Single-Cell Data},
    journal = {Cell},
    year = {2019},
    volume = {177},
    pages = {1888-1902},
    doi = {10.1016/j.cell.2019.05.031},
    url = {https://doi.org/10.1016/j.cell.2019.05.031},
  }

@Article{butler2018,
    author = {Andrew Butler and Paul Hoffman and Peter Smibert and
      Efthymia Papalexi and Rahul Satija},
    title = {Integrating single-cell transcriptomic data across
      different conditions, technologies, and species},
    journal = {Nature Biotechnology},
    year = {2018},
    volume = {36},
    pages = {411-420},
    doi = {10.1038/nbt.4096},
    url = {https://doi.org/10.1038/nbt.4096},
  }

@Article{satija2015,
    author = {Rahul Satija and Jeffrey A Farrell and David Gennert and
      Alexander F Schier and Aviv Regev},
    title = {Spatial reconstruction of single-cell gene expression
      data},
    journal = {Nature Biotechnology},
    year = {2015},
    volume = {33},
    pages = {495-502},
    doi = {10.1038/nbt.3192},
    url = {https://doi.org/10.1038/nbt.3192},
  }

@article{svensson2018spatialde,
  title={SpatialDE: identification of spatially variable genes},
  author={Svensson, Valentine and Teichmann, Sarah A and Stegle, Oliver},
  journal={Nature methods},
  volume={15},
  number={5},
  pages={343--346},
  year={2018},
  publisher={Nature Publishing Group US New York}
}

@article{sun2020statistical,
  title={Statistical analysis of spatial expression patterns for spatially resolved transcriptomic studies},
  author={Sun, Shiquan and Zhu, Jiaqiang and Zhou, Xiang},
  journal={Nature methods},
  volume={17},
  number={2},
  pages={193--200},
  year={2020},
  publisher={Nature Publishing Group US New York}
}

@article{zhu2021spark,
  title={SPARK-X: non-parametric modeling enables scalable and robust detection of spatial expression patterns for large spatial transcriptomic studies},
  author={Zhu, Jiaqiang and Sun, Shiquan and Zhou, Xiang},
  journal={Genome biology},
  volume={22},
  number={1},
  pages={184},
  year={2021},
  publisher={Springer}
}

@article{dries2021giotto,
  title={Giotto: a toolbox for integrative analysis and visualization of spatial expression data},
  author={Dries, Ruben and Zhu, Qian and Dong, Rui and Eng, Chee-Huat Linus and Li, Huipeng and Liu, Kan and Fu, Yuntian and Zhao, Tianxiao and Sarkar, Arpan and Bao, Feng and others},
  journal={Genome biology},
  volume={22},
  pages={1--31},
  year={2021},
  publisher={Springer}
}

@article{weber2023nnsvg,
  title={nnSVG for the scalable identification of spatially variable genes using nearest-neighbor Gaussian processes},
  author={Weber, Lukas M and Saha, Arkajyoti and Datta, Abhirup and Hansen, Kasper D and Hicks, Stephanie C},
  journal={Nature communications},
  volume={14},
  number={1},
  pages={4059},
  year={2023},
  publisher={Nature Publishing Group UK London}
}

@article{miller2021characterizing,
  title={Characterizing spatial gene expression heterogeneity in spatially resolved single-cell transcriptomic data with nonuniform cellular densities},
  author={Miller, Brendan F and Bambah-Mukku, Dhananjay and Dulac, Catherine and Zhuang, Xiaowei and Fan, Jean},
  journal={Genome research},
  volume={31},
  number={10},
  pages={1843--1855},
  year={2021},
  publisher={Cold Spring Harbor Lab}
}

@article{moran1950notes,
  title={Notes on continuous stochastic phenomena},
  author={Moran, Patrick AP},
  journal={Biometrika},
  volume={37},
  number={1/2},
  pages={17--23},
  year={1950},
  publisher={JSTOR}
}

@article{li2025spagrn,
  title={SpaGRN: investigating spatially informed regulatory paths for spatially resolved transcriptomics data},
  author={Li, Yao and Liu, Xiaobin and Guo, Lidong and Han, Kai and Fang, Shuangsang and Wan, Xinjiang and Wang, Dantong and Xu, Xun and Jiang, Ling and Fan, Guangyi and others},
  journal={Cell Systems},
  volume={16},
  number={4},
  year={2025},
  publisher={Elsevier}
}

@article{fang2023subcellular,
  title={Subcellular spatially resolved gene neighborhood networks in single cells},
  author={Fang, Zhou and Ford, Adam J and Hu, Thomas and Zhang, Nicholas and Mantalaris, Athanasios and Coskun, Ahmet F},
  journal={Cell Reports Methods},
  volume={3},
  number={5},
  year={2023},
  publisher={Elsevier}
}

@article{wang2022high,
  title={High-resolution 3D spatiotemporal transcriptomic maps of developing Drosophila embryos and larvae},
  author={Wang, Mingyue and Hu, Qinan and Lv, Tianhang and Wang, Yuhang and Lan, Qing and Xiang, Rong and Tu, Zhencheng and Wei, Yanrong and Han, Kai and Shi, Chang and others},
  journal={Developmental Cell},
  volume={57},
  number={10},
  pages={1271--1283},
  year={2022},
  publisher={Elsevier}
}

@article{detomaso2021hotspot,
  title={Hotspot identifies informative gene modules across modalities of single-cell genomics},
  author={DeTomaso, David and Yosef, Nir},
  journal={Cell systems},
  volume={12},
  number={5},
  pages={446--456},
  year={2021},
  publisher={Elsevier}
}

@article{wei2022spatial,
  title={Spatial charting of single-cell transcriptomes in tissues},
  author={Wei, Runmin and He, Siyuan and Bai, Shanshan and Sei, Emi and Hu, Min and Thompson, Alastair and Chen, Ken and Krishnamurthy, Savitri and Navin, Nicholas E},
  journal={Nature biotechnology},
  volume={40},
  number={8},
  pages={1190--1199},
  year={2022},
  publisher={Nature Publishing Group US New York}
}

@article{acharyya2022spacex,
  title={SpaceX: gene co-expression network estimation for spatial transcriptomics},
  author={Acharyya, Satwik and Zhou, Xiang and Baladandayuthapani, Veerabhadran},
  journal={Bioinformatics},
  volume={38},
  number={22},
  pages={5033--5041},
  year={2022},
  publisher={Oxford University Press}
}

@article{fernandez2015brain,
  title={Brain expressed and X-linked (Bex) proteins are intrinsically disordered proteins (IDPs) and form new signaling hubs},
  author={Fernandez, Eva M and D{\'\i}az-Ceso, Mar{\'\i}a D and Vilar, Mar{\c{c}}al},
  journal={PloS one},
  volume={10},
  number={1},
  pages={e0117206},
  year={2015},
  publisher={Public Library of Science San Francisco, CA USA}
}

@article{martinotti2015emerging,
  title={Emerging roles for HMGB1 protein in immunity, inflammation, and cancer},
  author={Martinotti, Simona and Patrone, Mauro and Ranzato, Elia},
  journal={ImmunoTargets and therapy},
  pages={101--109},
  year={2015},
  publisher={Taylor \& Francis}
}

@article{yuan2020high,
  title={High mobility group box 1 (HMGB1): a pivotal regulator of hematopoietic malignancies},
  author={Yuan, Shunling and Liu, Zhaoping and Xu, Zhenru and Liu, Jing and Zhang, Ji},
  journal={Journal of Hematology \& Oncology},
  volume={13},
  number={1},
  pages={91},
  year={2020},
  publisher={Springer}
}

@article{guan2023comprehensive,
  title={The comprehensive role of high mobility group box 1 (HMGB1) protein in different tumors: a pan-cancer analysis},
  author={Guan, Hui and Zhong, Ming and Ma, Kongyang and Tang, Chun and Wang, Xiaohua and Ouyang, Muzi and Qin, Rencai and Chen, Jiasi and Zhu, Enyi and Zhu, Ting and others},
  journal={Journal of inflammation research},
  pages={617--637},
  year={2023},
  publisher={Taylor \& Francis}
}

@article{ruggieri2024hmgb1,
  title={HMGB1, an evolving pleiotropic protein critical for cellular and tissue homeostasis: role in aging and age-related diseases},
  author={Ruggieri, Elena and Di Domenico, Erika and Locatelli, Andrea Giacomo and Isopo, Flavio and Damanti, Sarah and De Lorenzo, Rebecca and Milan, Enrico and Musco, Giovanna and Rovere-Querini, Patrizia and Cenci, Simone and others},
  journal={Ageing Research Reviews},
  volume={102},
  pages={102550},
  year={2024},
  publisher={Elsevier}
}

@article{song2021role,
  title={The role of neuropeptide somatostatin in the brain and its application in treating neurological disorders},
  author={Song, You-Hyang and Yoon, Jiwon and Lee, Seung-Hee},
  journal={Experimental \& Molecular Medicine},
  volume={53},
  number={3},
  pages={328--338},
  year={2021},
  publisher={Nature Publishing Group UK London}
}

@article{alfimova2021relationship,
  title={Relationship between DNA Methylation within the YJEFN3 Gene and Cognitive Deficit in Schizophrenia Spectrum Disorders},
  author={Alfimova, MV and Kondratyev, NV and Golov, AK and Kaleda, VG and Abramova, LI and Golimbet, VE},
  journal={Russian Journal of Genetics},
  volume={57},
  number={9},
  pages={1092--1099},
  year={2021},
  publisher={Springer}
}

@misc{zhou2023role,
  title={The role of CXCL family members in different diseases. Cell Death Discov. 9, 212},
  author={Zhou, C and Gao, Y and Ding, P and Wu, T and Ji, G},
  year={2023}
}

@article{westrich2020multifarious,
  title={The multifarious roles of the chemokine CXCL14 in cancer progression and immune responses},
  author={Westrich, Joseph A and Vermeer, Daniel W and Colbert, Paul L and Spanos, William C and Pyeon, Dohun},
  journal={Molecular carcinogenesis},
  volume={59},
  number={7},
  pages={794--806},
  year={2020},
  publisher={Wiley Online Library}
}

@article{yang2015glial,
  title={Glial fibrillary acidic protein: from intermediate filament assembly and gliosis to neurobiomarker},
  author={Yang, Zhihui and Wang, Kevin KW},
  journal={Trends in neurosciences},
  volume={38},
  number={6},
  pages={364--374},
  year={2015},
  publisher={Elsevier}
}

@article{quinlan2007gfap,
  title={GFAP and its role in Alexander disease},
  author={Quinlan, Roy A and Brenner, Michael and Goldman, James E and Messing, Albee},
  journal={Experimental cell research},
  volume={313},
  number={10},
  pages={2077--2087},
  year={2007},
  publisher={Elsevier}
}

@article{liu2022joint,
  title={Joint dimension reduction and clustering analysis of single-cell RNA-seq and spatial transcriptomics data},
  author={Liu, Wei and Liao, Xu and Yang, Yi and Lin, Huazhen and Yeong, Joe and Zhou, Xiang and Shi, Xingjie and Liu, Jin},
  journal={Nucleic acids research},
  volume={50},
  number={12},
  pages={e72--e72},
  year={2022},
  publisher={Oxford University Press}
}

\newpage
\appendix
\section{Appendix}\label{Appendix}
\subsection{Proof of Theorem~\ref{thm:BM}} To prove Theorem~\ref{thm:BM}, we need the following lemma.
\begin{lemma}\label{lemma:HolderCont}
Suppose $G_0(x) = (g^{0}_{ij}(x))_{p \times q}$ be any matrix that belongs to some class defined as $\mathbb{R}^{p \times q}_{G_0}$. $G_0(x)$ depends on univariate $x\in(a,b]$ where $-\infty < a < b < \infty$. For all $i \in \{1, \cdots, p\}$ and $j \in \{1, \cdots, q\}$, assume $g^{0}_{ij}(x) \in \mathcal{H}^{\nu}$, where $\mathcal{H}^{\nu}$ is the space of uniformly $\nu\mbox{-}$H$\Ddot{o}$lder continuous functions with $\nu\mbox{-}$H$\Ddot{o}$lder coefficient $\|f\|_{\mathcal{H}^\nu}$ defined as, 
\begin{align*}
    \mathcal{H}^\nu = \left\{ f : [a,b] \to \mathbbm{R} : \|f\|_{\mathcal{H}^\nu} = \sup_{x, x^\prime \in [a, b]} \frac{\left|f(x) - f(x^\prime)\right|}{|x-x^\prime|^\nu} < \infty  \right\}.
\end{align*}
Then, for any $\epsilon>0$ there exists $K^\ast \in \mathbb{N}$ depending on $\epsilon$, a set of matrices $\{G_{0k}\}_{k = 1}^{K^\ast}$ of dimension $p \times q$ where $G_{0k} = (g^{0k}_{ij})_{p \times q} \in \mathbb{R}^{p \times q}_{G_0}$ and a set of positive real numbers $a  = a_1 \leq a_2 < \cdots < a_{K^\ast} \leq a_{K^\ast+1} = b$ such that 
$$ \left\| G_0(\cdot) - \sum_{k = 1}^{K^\ast} \mathrm{1}_{(a_k, a_{k+1}]}(\cdot)G_{0k} \right\|_1 < \epsilon, $$
where $G_{0k}$ and $a_k$ depend on $G_0(\cdot)$ and $\epsilon$ and $\left\|A(\cdot) \right\|_1 = \int_{a}^b \left\| A(x) \right\| dx$ for any matrix norm.
\end{lemma}

\begin{proof}
For any given $\epsilon>0$, we choose $K^\ast = \left[ \left( \frac{pqH^2_c(b-a)^{2\nu+2}}{\epsilon^2} \right)^{1/2\nu}\right] \in \mathbb{N}$ where $H_c = \operatorname{max}_{i,j}\left\| g^0_{ij}\right\|_{\mathcal{H}^\nu}$ with $\left\|g^{0}_{ij}\right\|_{\mathcal{H}^\nu}$ being the $\nu\mbox{-}$H$\Ddot{o}$lder coefficient.
Further, we choose a piecewise constant function, $$g_{ij}(x) = \sum_{k = 1}^{K^\ast}\delta^k_{ij}\mathrm{1}_{\Delta_k} (x), \hspace{5pt} \text{where} \hspace{5pt} a_k = a+\frac{(b-a)(k-1)}{K^\ast} \hspace{5pt} \text{and} \hspace{5pt} \Delta_k = (a_k, a_{k+1}].$$ 
Then for any $k \in \{1, \cdots, K^\ast\}$ and $x \in \Delta_k$, by choosing $\delta^k_{ij} = g^{0}_{ij}(a_k)$ we have,
\begin{align*}
    \left| g^{0}_{ij}(x) - g_{ij}(x) \right| = \left| g^{0}_{ij}(x) - \delta^k_{ij} \right| = \left| g^{0}_{ij}(x) - g^{0}_{ij}(a_k) \right| \leq \left\|g^{0}_{ij}\right\|_{\mathcal{H}^\nu}\frac{(b-a)^\nu}{{K^\ast}^\nu},
\end{align*}
using $\nu\mbox{-}$H$\Ddot{o}$lder inequality. 

Therefore the following holds,
\begin{align*}
    \begin{split}
        \left\| G_0(\cdot) - \sum_{k = 1}^{K^\ast} \mathrm{1}_{(a_k, a_{k+1}]} (\cdot)G_{0k} \right\|_1 & = \int_{a}^b \left\| G_0(x) - \sum_{k = 1}^{K^\ast} \mathrm{1}_{(a_k, a_{k+1}]} (x)G_{0k} \right\| \, dx\\
        & \leq \sqrt{\int_{a}^b 1 \, dx} \sqrt{\int_{a}^b \left\| G_0(x) - \sum_{k = 1}^{K^\ast} \mathrm{1}_{(a_k, a_{k+1}]} (x)G_{0k} \right\|^2 \, dx}\\
        & = \sqrt{b-a} \sqrt{\int_{a}^b \left\| G_0(x) - \sum_{k = 1}^{K^\ast} \mathrm{1}_{(a_k, a_{k+1}]} (x)G_{0k} \right\|^2 \, dx}\\
        & = \sqrt{b-a} \sqrt{\sum_{i = 1}^p \sum_{j = 1}^q \int_{a}^b \left(g^0_{ij}(x) - \sum_{k = 1}^{K^\ast} \mathrm{1}_{\Delta_k}(x)g^{0k}_{ij} \right)^2}\, dx\\
        & = \sqrt{b-a} \sqrt{\sum_{i = 1}^p \sum_{j = 1}^q \int_{a}^b \left(g^0_{ij}(x) - g_{ij}(x) \right)^2}\, dx\\
        & \leq \sqrt{\frac{pq H^2_c(b -a)^{2\nu+2}}{{K^\ast}^{2\nu}}} = \epsilon.
    \end{split}
\end{align*}
Here the first inequality follows using Cauchy Schwarz for definite integrals, the third equality follows by choosing the matrix norm to be the Frobenius norm, the fourth equality follows by setting $g^{0}_{ij} = \delta^k_{ij}$ since $g_{ij}(x) = \sum_{k = 1}^{K^\ast} \delta^k_{ij} \mathrm{1}_{\Delta_k}(x)$, and the last inequality follows from the choice of $K^\ast$, which depends on $\epsilon$.
\end{proof}

With Lemma~\ref{lemma:HolderCont}, we proceed to prove Theorem~\ref{thm:BM}, which we restate here.

\textbf{Theorem~\ref{thm:BM}.} {\it Suppose $(B_0(\cdot))_{p \times p}$ be any matrix of SEM coefficients and $(M_0(\cdot))_{p \times 1}$ be the intercept vector 
as in Equation~\eqref{eq:sem0},
 where $B_0(\cdot)$, $M_0(\cdot)$ are continuous functions for univariate $x \in (a,b]$ with $-\infty<a<b<\infty$. We assume that all the entries of $B_0(\cdot)$ and $M_0(\cdot)$ belong to the space of uniformly $\nu\mbox{-}$H$\Ddot{o}$lder continuous functions. Then for any $\epsilon>0$, there exists $K^* \in \mathbb{N}$ such that 
 \begin{itemize}
 \item [(i)] For any absolutely continuous prior distributions on $(B_j, \mu_j, \lambda_j) \in (\mathbb{S}_p, \mathbb{R}, \mathbb{R}^{+})$ and on \\
 \noindent $(\pi_1, \cdots, \pi_{K^\ast})$ we have,
\begin{align*}
\begin{split}
     &\mathbb{P}\left(\left\| B_0 -\sum_{j = 1}^{K^*} \pi_j(x)B_j \right\|_1 < \epsilon \right) > 0\\
 \end{split}
\end{align*}
where $\mathbb{S}_p$ be the set of $p \times p$ stable matrices and $\pi_j(x) = \frac{\pi_j N \left(x; \mu_j, \lambda_j \right)}{\sum_{j = 1}^{K^\ast} \pi_j N \left(x; \mu_j, \lambda_j \right)}$ for $j \in \{1, \cdots, K^*\}$. 
 
 \item [(ii)] Similarly, for any absolutely continuous prior distributions on $(M_j, \mu_j, \lambda_j) \in (\mathbb{R}^{p}, \mathbb{R}, \mathbb{R^+})$ and on $(\pi_1, \cdots, \pi_{K^\ast})$ we have,
\begin{align*}
 \begin{split}
     &\mathbb{P}\left(\left\| M_0 -\sum_{j = 1}^{K^*} \pi_j(x)M_j \right\|_1 < \epsilon \right) > 0\\
 \end{split}
\end{align*}
\end{itemize}
where $\pi_j(x)$ is of the same form as part $(i)$ of the theorem.}

\begin{proof} 

\begin{itemize}
    \item [$(i)$]  
 
 Assume that the interval $(a,b]$ are first subdivided into $K$ many intervals of length $L = \frac{(b-a)}{L}$. Each of the intervals are further subdivided into $K^\prime$ many subintervals of length $\frac{L}{K^\prime}$. Hence, we divide the interval $(a,b]$ into $K^\ast = KK^\prime$ subintervals.
 
 By using the triangle inequality, we have:
\begin{align*}
\begin{split}
      \left\| B_0(\cdot) - \sum_{k = 1}^{K} \sum_{l = 1}^{K^\prime} \pi_{kl} (\cdot) B_{kl} \right\|_1 \leq T_1 + T_2 + T_3,\\
\end{split}
\end{align*}
where, 
\begin{align*}
    \begin{split}
        & T_1 = \left\| B_0(\cdot) - \sum_{k = 1}^{K} \sum_{l = 1}^{K^\prime} \mathrm{1}_{(a_{kl, ak(l+1)}]} (\cdot) B_{0kl} \right\|_1,\\
        & T_2 = \left\| \sum_{k = 1}^{K} \sum_{l = 1}^{K^\prime} \mathrm{1}_{(a_{kl}, a_{k(l+1)]}} (\cdot) B_{0kl} - \sum_{k = 1}^{K} \sum_{l = 1}^{K^\prime} \pi_{kl} (\cdot) B_{0kl} \right\|_1, \\
        & T_3 = \left\| \sum_{k = 1}^{K} \sum_{l = 1}^{K^\prime} \pi_{kl} (\cdot) B_{0kl} - \sum_{k = 1}^{K} \sum_{l = 1}^{K^\prime} \pi_{kl} (\cdot) B_{kl} \right\|_1.
    \end{split}
\end{align*}

\underline{Bound $T_1$}. Using Lemma~\ref{lemma:HolderCont}, consider $G_0(\cdot) = B_0(\cdot)$, $G_{0kl} = B_{0kl}$ for $k = 1, \cdots, K$ and $l = 1, \cdots, K^\prime$ and $\mathbbm{R}^{p\times q}_{G_0} = \mathbb{S}_p$, set of stable matrices of size $p$. Therefore, given any $\epsilon>0$, there exists $K^\ast = KK^\prime \in \mathbbm{N}$ depending on $\epsilon$, a set of stable matrices $\{{B_{0kl}\}_{1\leq k\leq K,\  1\leq l\leq K^\prime}}$ of dimension $p \times p$ and a set of positive real numbers in ascending order $\{{a_{kl}\}_{1\leq k\leq K,\  1\leq l\leq K^\prime}}$ such that,
 \begin{align}\label{eq: B0_B0k}
    T_1 = \left\| B_0(\cdot) - \sum_{k = 1}^K \sum_{l = 1}^{K^\prime} \mathrm{1}_{(a_{kl}, a_{k(l+1)}]}(\cdot)B_{0kl} \right\|_1 < \frac{\epsilon}{3}.
 \end{align}

\underline{Bound $T_2$.} Consider the term $\pi_{kl}(x)$. For all $k = 1, \cdots, K$ and $l = 1, \cdots, K^\prime$, with the choice of $\lambda_{kl}= \lambda$ and $\pi_{kl} = \pi$, we have the following,
\begin{align*}
\begin{split}
      \pi_{kl}(x) & = \frac{\pi_{kl}N(x; \mu_{kl}, \lambda_{kl})}{\sum_{u = 1}^{K} \sum_{v = 1}^{K^\prime} \pi_{uv}N(x; \mu_{uv}, \lambda_{uv})} = \frac{\exp\left\{ -\frac{1}{2} \frac{(x - \mu_{kl})^2}{\lambda} \right\}}{\sum_{u = 1}^{K} \sum_{v = 1}^{K^\prime}  \exp\left\{ -\frac{1}{2} \frac{(x - \mu_{uv})^2}{\lambda} \right\}}.\\
\end{split}
\end{align*}

In order to establish that $\pi_{kl}(\cdot)$ approximates $\mathrm{1}_{(a_{kl}, a_{k(l+1)}]}(\cdot)$ well in $\|\cdot\|_1$, we choose to focus on appropriate neighborhoods of $\mu_{kl}$ and $\lambda$ that have positive probabilities. In particular, if we choose,
$$ \mu_{kl} \in \left[ \frac{a_{kl} + a_{k(l+1)}}{2} - \frac{L}{6{K^\prime}^2}, \frac{a_{kl} + a_{k(l+1)}}{2} + \frac{L}{6{K^\prime}^2}  \right] \hspace{5pt} \text{and} \hspace{5pt} \lambda \in \left[ \frac{L^2}{2{K^\prime}^4}, \frac{L^2}{{K^\prime}^4} \right], $$ the prior probabilities of these neighborhoods are non-zero because of absolutely continuous priors.
Further, depending on $L$, any given $\epsilon$ and corresponding $K$ and \\
\noindent $\zeta = \max_{1 \leq k \leq K, 1 \leq l \leq K^\prime} \left\{\|B_{0kl}\|\right\}$, we consider $K^\prime = \psi(\epsilon, \zeta, L, K)$. We mention the specific choice of $K^\prime$ in the subsequent part of the proof.
To approximate an indicator function by $\pi_{kl}(x)$ we consider the term, $\|\pi_{kl}(x) - \mathrm{1}_{(a_{kl}, a_{k(l+1)}]}(x)\|_1$ which can be written as follows,
\begin{align*}
    \begin{split}
        & \|\pi_{kl}(x) - \mathrm{1}_{(a_{kl}, a_{k(l+1)}]}(x)\|_1 \\
        & = \int_{a}^b \left| \pi_{kl}(x) - \mathrm{1}_{(a_{kl}, a_{k(l+1)}]}(x) \right|\, dx\\
        & = \int_{x \in (a_{kl}, a_{k(l+1)}]} \left| \pi_{kl}(x) - \mathrm{1}_{(a_{kl}, a_{k(l+1)}]}(x) \right|\, dx + \int_{x \notin (a_{kl}, a_{k(l+1)}]} \left| \pi_{kl}(x) - \mathrm{1}_{(a_{kl}, a_{k(l+1)}]}(x) \right|\, dx.
    \end{split}
\end{align*}

Furthermore, $\int_{x \in (a_{kl}, a_{k(l+1)}]} \left| \pi_{kl}(x) - \mathrm{1}_{(a_{kl}, a_{k(l+1)}]}(x) \right|\, dx = I^*_1 + I^*_2 + I^*_3$ where 
$$I^*_1 = \int_{x \in (a_{kl}, a_{kl} + \frac{L}{5{K^\prime}^2}]} \left| \pi_{kl}(x) - 1 \right|\, dx, $$
$$I^*_2 = \int_{x \in (a_{kl} + \frac{L}{5{K^\prime}^2}, a_{k(l+1)} - \frac{L}{5{K^\prime}^2}]} \left| \pi_{kl}(x) - 1 \right|\, dx, $$
$$I^*_3 = \int_{x \in (a_{k(l+1)} - \frac{L}{5{K^\prime}^2}, a_{k(l+1)}]} \left| \pi_{kl}(x) - 1 \right|\, dx.$$

Note that, $I^*_1 \leq  \int_{x \in (a_{kl}, a_{kl} + \frac{L}{5{K^\prime}^2}]} 1\, dx \leq \frac{L}{5{K^\prime}^2}.$ Similarly we also have, $I^*_3 \leq \frac{L}{5{K^\prime}^2}$. Next we bound $I^*_2$ as follows,
\begin{align*}
\begin{split}
      I^*_2 & = \int_{x \in (a_{kl} + \frac{L}{5{K^\prime}^2}, a_{k(l+1)} - \frac{L}{5{K^\prime}^2}]} \left| \pi_{kl}(x) - 1 \right|\, dx\\
      & = \int_{x \in (a_{kl} + \frac{L}{5{K^\prime}^2}, a_{k(l+1)} - \frac{L}{5{K^\prime}^2}]} \left| \frac{\exp\left\{ -\frac{1}{2} \frac{(x - \mu_{kl})^2}{\lambda} \right\}}{\sum_{u = 1}^{K} \sum_{v = 1}^{K^\prime}  \exp\left\{ -\frac{1}{2} \frac{(x - \mu_{uv})^2}{\lambda} \right\}} - 1 \right|\, dx\\
      & = \int_{x \in (a_{kl} + \frac{L}{5{K^\prime}^2}, a_{k(l+1)} - \frac{L}{5{K^\prime}^2}]} \left| \frac{1}{1+\sum \sum_{(u,v) \neq (k,l)} \exp\left\{ -\frac{1}{2} \frac{(x - \mu_{uv})^2}{\lambda} + \frac{1}{2} \frac{(x - \mu_{kl})^2}{\lambda}  \right\}} - 1 \right|\, dx\\
      & \leq \int_{x \in (a_{kl} + \frac{L}{5{K^\prime}^2}, a_{k(l+1)} - \frac{L}{5{K^\prime}^2}]} {\sum \sum}_{(u,v) \neq (k,l)} \exp\left\{ -\frac{1}{2} \frac{(x - \mu_{uv})^2}{\lambda} + \frac{1}{2} \frac{(x - \mu_{kl})^2}{\lambda}  \right\}\, dx\\
      & \leq \frac{L}{K^\prime} \cdot KK^\prime \cdot \exp\left\{- \frac{\left(\frac{L}{2K^\prime} - \frac{L}{6{K^\prime}^2} + \frac{L}{5{K^\prime}^2} \right)^2}{\frac{2L^2}{{K^\prime}^4}} + \frac{\left(\frac{L}{2K^\prime} + \frac{L}{6{K^\prime}^2} - \frac{L}{5{K^\prime}^2} \right)^2}{\frac{2L^2}{{K^\prime}^4}} \right\}\\
      & \leq LK \exp \left\{ - \frac{K^\prime}{30} \right\}. 
\end{split}
\end{align*}

Finally, we proceed to bound $\int_{x \notin (a_{kl}, a_{k(l+1)}]} \left| \pi_{kl}(x) - \mathrm{1}_{(a_{kl}, a_{k(l+1)}]}(x) \right|\, dx$. Define the integral to be $I^\ast_4$ and it can be written as,
\begin{align*}
    \begin{split}
       I^\ast_4 = & \int_{x \notin (a_{kl}, a_{k(l+1)}]} \left| \pi_{kl}(x) - \mathrm{1}_{(a_{kl}, a_{k(l+1)}]}(x) \right|\, dx \\
        & = \int_{x \notin (a_{kl}, a_{k(l+1)}]} \left| \frac{\exp\left\{ -\frac{1}{2} \frac{(x - \mu_{kl})^2}{\lambda} \right\}}{\sum_{u = 1}^{K} \sum_{v = 1}^{K^\prime}  \exp\left\{ -\frac{1}{2} \frac{(x - \mu_{uv})^2}{\lambda} \right\}}\right|\, dx \\
        & = \int_{x \notin (a_{kl}, a_{k(l+1)}]} \left| \frac{1}{\sum_{u = 1}^K \sum_{v = 1}^{K^\prime} \exp\left\{ -\frac{1}{2} \frac{(x - \mu_{uv})^2}{\lambda} + \frac{1}{2} \frac{(x - \mu_{kl})^2}{\lambda}  \right\}} \right|\, dx
    \end{split}
\end{align*}

For notational simplicity, we define $C(x) = \sum_{u = 1}^K \sum_{v = 1}^{K^\prime} \exp\left\{ -\frac{1}{2} \frac{(x - \mu_{uv})^2}{\lambda} + \frac{1}{2} \frac{(x - \mu_{kl})^2}{\lambda}  \right\}$. Therefore, we have 
\begin{align}\label{eq: Obj_x}
    I^\ast_4 = \int_{x \notin (a_{kl}, a_{k(l+1)}]} \left| \pi_{kl}(x) - \mathrm{1}_{(a_{kl}, a_{k(l+1)}]}(x) \right|\, dx =   \int_{x \notin (a_{kl}, a_{k(l+1)}]} \frac{1}{C(x)} \, dx.
\end{align}

Now define the following,
\begin{align*}
    L_{kl} = \left\{ (u_L, v_L) \in \{1, \cdots, k-1\} \times \{1, \cdots, K^\prime\} \cup \{k\} \times \{1, \cdots, l-2\} \right\},\\
    R_{kl} = \left\{ (u_R, v_R) \in \{k\} \times \{l+2, \cdots, K^\prime\}\cup \{k+1, \cdots, K\} \times \{1, \cdots, K^\prime\} \right\}.
\end{align*}
This means $L_{kl}$ is the collection of the subscripts of the left ends of all subintervals (of length $\frac{L}{K^\prime}$) to the left of $a_{k(l-1)}$. Similarly, $R_{kl}$ is to the right of $a_{k(l+2)}$. We further subdivide the interval $(a,b] \setminus (a_{kl}, a_{k(l+1)}]$ into four parts and write,
\begin{align*}
    &(a,b] \setminus (a_{kl}, a_{k(l+1)}]\\
    &= \left\{ \cup_{(u_L, v_L) \in L_{kl}} (a_{u_L v_L}, a_{u_L (v_L + 1)}]  \right\} \cup (a_{k(l-1)}, a_{kl}] \cup (a_{k(l+1)}, a_{k(l+2)}] \cup \\ & \hspace{10cm} \left\{\cup_{(u_R, v_R) \in R_{kl}} (a_{u_R v_R, a_{u_R (v_R+1)}}]  \right\}, 
\end{align*}
where $u_L, v_L$ are the indices based on which intervals in $L_{kl}$ are constructed. In a similar way, $u_R, v_R$ are also interpreted for $R_{kl}$.

In particular, from Equation~\eqref{eq: Obj_x}, we can write the integral as,
\begin{align*}
    \begin{split}
        & \int_{x \notin (a_{kl}, a_{k(l+1)}]} \left| \pi_{kl}(x) - \mathrm{1}_{(a_{kl}, a_{k(l+1)}]}(x) \right|\, dx = I^\ast_{41} + I^\ast_{42} + I^\ast_{43} + I^\ast_{44},\\
    \end{split}
\end{align*}
where, $$ I^\ast_{41}  = \sum_{(u_L, v_L] \in L_{kl}} \int_{x \in (a_{u_L v_L}, a_{u_L (v_L +1)}]} \frac{1}{C(x)} \, dx,$$ $$ I^\ast_{42} =  \int_{x \in (a_{k(l-1)}, a_{kl}]} \frac{1}{C(x)} \, dx, $$ $$ I^\ast_{43} = \int_{x \in (a_{k(l+1)}, a_{k(l+2)}]} \frac{1}{C(x)} \, dx, $$
$$ I^\ast_{44} = \sum_{(u_R, v_R] \in R_{kl}} \int_{x \in (a_{u_R v_R}, a_{u_R (v_R +1)}]} \frac{1}{C(x)} \, dx.  $$

Next we consider each of the integrals separately in the following. We start with $I^\ast_{41}$.

Consider $x \in (a_{u_L v_L}, a_{u_L (v_L + 1)}]$ for some $(u_L, v_L) \in L_{kl}$. Then we can write the following lower bound on $C(x)$ given by, 
\begin{align*}
    \begin{split}
        C(x) & = \sum_{u = 1}^K \sum_{v = 1}^{K^\prime} \exp \left[-\frac{1}{2\lambda} \left\{ (x - \mu_{uv})^2 - (x - \mu_{kl})^2 \right\}\right] \\
        & \geq \exp \left[-\frac{1}{2\lambda} \left\{ (x - \mu_{u_L v_L})^2 - (x - \mu_{kl})^2 \right\}\right]\\
        & \geq \exp \left[-\frac{1}{2\lambda} \left\{ (x - \mu_{u_L v_L})^2 - (a_{k(l-1)} - \mu_{kl})^2 \right\}\right]\\
        & \geq \exp \left[\frac{{K^\prime}^4}{2L^2} \left\{ \left(\frac{3L}{2K^\prime} - \frac{L}{6{K^\prime}^2} \right)^2 - \left(\frac{L}{2K^\prime}+\frac{L}{6{K^\prime}^2} \right)^2 \right\}\right]\\
        & = \exp \left[ {K^\prime}^2 - \frac{1}{3}K^\prime \right] \geq \exp \left[  \frac{1}{3}K^\prime \right].
    \end{split}
\end{align*}
\end{itemize}
Therefore, using this bound we have the following,
\begin{align}\label{eq: I41}
    \begin{split}
        I^\ast_{41} & = \sum_{(u_L, v_L] \in L_{kl}} \int_{x \in (a_{u_L v_L}, a_{u_L (v_L +1)}]} \frac{1}{C(x)} \, dx\\
        & \leq \sum_{(u_L, v_L] \in L_{kl}} \int_{x \in (a_{u_L v_L}, a_{u_L (v_L +1)}]} \exp \left[ -\frac{1}{3}K^\prime \right] \, dx.\\
    \end{split}
\end{align}

The bound for $I^\ast_{44}$ is the same as that for $I^\ast_{41}$ because of the symmetry. Consider $x \in (a_{u_R v_R}, a_{u_R (v_R + 1)}]$ for some indices of the form $(u_R, v_R) \in R_{kl}$. 
Then $C(x)$ is again lower bounded as $C(x) \geq \exp\left[\frac{1}{3}K^\prime \right]$. Therefore we have, 
\begin{align}\label{eq: I44}
    \begin{split}
        I^\ast_{44} & = \sum_{(u_R, v_R] \in L_{kl}} \int_{x \in (a_{u_R v_R}, a_{u_R (v_R +1)}]} \frac{1}{C(x)} \, dx\\
        & \leq \sum_{(u_R, v_R] \in L_{kl}} \int_{x \in (a_{u_R v_R}, a_{u_R (v_R +1)}]} \exp \left[ -\frac{1}{3}K^\prime \right] \, dx.\\ 
    \end{split}
\end{align}

Now we consider $I^\ast_{42}$ and further subdivide $(a_{k(l-1)},a_{kl}] = (a_{k(l-1)},a_{k(l-1)^\ast}]  \cup (a_{k(l-1)^\ast},a_{kl^\ast}]  \cup (a_{kl^\ast},a_{kl}]$ where $a_{k(l-1)^\ast} = a_{k(l-1)} + \frac{L}{5{K^\prime}^2}$ and $a_{kl^\ast} = a_{kl} - \frac{L}{5{K^\prime}^2}$.

When $x \in (a_{k(l-1)^\ast},a_{kl^\ast}]$, similarly as before we maximize the term $C(x)$ as,
\begin{align*}
    \begin{split}
         C(x) & = \sum_{u = 1}^K \sum_{v = 1}^{K^\prime} \exp \left[-\frac{1}{2\lambda} \left\{ (x - \mu_{uv})^2 - (x - \mu_{kl})^2 \right\}\right] \\
        & \geq \exp \left[-\frac{1}{2\lambda} \left\{ (x - \mu_{k (l-1)})^2 - (x - \mu_{kl})^2 \right\}\right]\\
        & \geq \exp \left[-\frac{1}{2\lambda} \left\{ (a_{k(l-1)^\ast} - \mu_{k(l-1)})^2 - (a_{kl^\ast} - \mu_{kl})^2 \right\}\right]\\
        & \geq \exp \left[\frac{{K^\prime}^4}{2L^2} \left\{ \left( \frac{L}{2K^\prime} + \frac{L}{5{K^\prime}^2} - \frac{L}{6{K^\prime}^2} \right)^2 - \left( \frac{L}{2K^\prime} - \frac{L}{5{K^\prime}^2} + \frac{L}{6{K^\prime}^2}\right)^2   \right\} \right]\\
        & \geq \exp\left[\frac{1}{30}K^\prime \right].
    \end{split}
\end{align*}

Therefore, using $\frac{1}{C(x)} \leq \exp\left[-\frac{1}{30}K^\prime \right]$ when $x \in (a_{k(l-1)^\ast},a_{kl^\ast}]$ and $\frac{1}{C(x)} \leq 1$ when $x \in (a_{k(l-1)},a_{k(l-1)^\ast}]  \cup (a_{kl^\ast},a_{kl}]$ we can write the integral $I^\ast_{42}$ as,
\begin{align}\label{eq: I42}
    \begin{split}
        I^\ast_{42} & \leq \int_{x \in (a_{k(l-1)},a_{k(l-1)^\ast}]} \frac{1}{C(x)}\, dx \, + \int_{x \in (a_{k(l-1)^\ast},a_{kl^\ast}]} \frac{1}{C(x)}\, dx \, + \int_{x \in (a_{kl^\ast},a_{kl}]} \frac{1}{C(x)}\, dx \\
        & \leq \int_{x \in (a_{k(l-1)},a_{k(l-1)^\ast}]} \, dx \, + \int_{x \in (a_{k(l-1)^\ast},a_{kl^\ast}]} \exp\left[-\frac{1}{30}K^\prime \right]  dx \, + \int_{x \in (a_{kl^\ast},a_{kl}]}  dx\\
        & \leq \frac{L}{K^\prime}\exp\left[-\frac{1}{30}K^\prime \right] + \frac{2L}{5{K^\prime}^2}.
    \end{split}
\end{align}

Now $(a_{k(l+1)}, a_{k(l+2)}]$ being the opposite interval of $(a_{k(l-1)}, a_{kl}]$, in an exact similar way we can bound $I^\ast_{43}$. Therefore, we can write the following,
\begin{align}\label{eq: I43}
    I^\ast_{43} \leq \int_{x \in (a_{k(l+1)}, a_{k(l+2)}]} \frac{1}{C(x)} \, dx \leq \frac{L}{K^\prime}\exp\left[-\frac{1}{30}K^\prime \right] + \frac{2L}{5{K^\prime}^2}.
\end{align}

Combining Equations~\eqref{eq: I41}-\eqref{eq: I43}, we obtain the following,
\begin{align}
\begin{split}
    I^\ast_4
    & = I^\ast_{41} + I^\ast_{42} + I^\ast_{43} + I^\ast_{44} \\
    & \leq  \sum_{(u_L, v_L] \in L_{kl}} \int_{x \in (a_{u_L v_L}, a_{u_L (v_L +1)}]} \exp \left[ -\frac{1}{3}K^\prime \right] \, dx\, + \frac{2L}{K^\prime}\exp\left[-\frac{1}{30}K^\prime \right] \\ & \hspace{1cm} + \frac{4L}{5{K^\prime}^2}\, + \sum_{(u_R, v_R] \in L_{kl}} \int_{x \in (a_{u_R v_R}, a_{u_R (v_R +1)}]} \exp \left[ -\frac{1}{3}K^\prime \right] \, dx\\
    & \leq \exp \left[ -\frac{1}{3}K^\prime \right]\int_{a}^b dx + \frac{2L}{K^\prime}\exp\left[-\frac{1}{30}K^\prime \right] + \frac{4L}{5{K^\prime}^2}\\
    & \leq LK \exp \left[ -\frac{1}{3}K^\prime \right]  + \frac{2L}{K^\prime}\exp\left[-\frac{1}{30}K^\prime \right] + \frac{4L}{5{K^\prime}^2}\\
    & \leq 3LK \exp \left[ -\frac{1}{30}K^\prime \right] + \frac{4L}{5{K^\prime}^2}.\\
\end{split}
\end{align}

In order to have the upper bound for $LK \exp \left[ -\frac{1}{30}K^\prime \right] + \frac{4L}{5{K^\prime}^2}$, we choose $K^\prime$ in the following way. We know $\lim_{K^\prime \to \infty} K^\prime \exp\left[ -\frac{K^\prime}{30}\right] = 0$. That is, there exists $K^\prime_0(L,K,\zeta,\epsilon)$ such that for all $K^\prime \geq K^\prime_0(L,K,\zeta,\epsilon)$, $3LK\exp\left[ -\frac{K^\prime}{30}\right] \leq \frac{\epsilon}{24\zeta K K^\prime}$. Moreover we would also like to bound $\frac{4L}{5{K^\prime}^2}$ with $\frac{\epsilon}{24\zeta K K^\prime}$. Hence, the specific choice of $K^\prime$ is given by $K^\prime = \left[\max \left\{\frac{96LK\zeta}{5\epsilon}, K^\prime_0(L,K,\zeta,\epsilon) \right\}\right]+1$. Under this choice, we have the following,
$$ I^\ast_{1} \leq \frac{L}{5{K^\prime}^2} \leq \frac{\epsilon}{12\zeta K K^\prime}, $$
$$ I^\ast_{2} \leq LK \exp \left[-\frac{K^\prime}{30} \right] \leq \frac{\epsilon}{12\zeta K K^\prime}, $$
$$I^\ast_{3} \leq \frac{L}{5{K^\prime}^2} \leq \frac{\epsilon}{12\zeta K K^\prime},$$
$$ I^\ast_{4} \leq 3LK\exp\left[-\frac{K^\prime}{30} \right] +  \frac{4L}{5{K^\prime}^2} \leq \frac{\epsilon}{12\zeta K K^\prime}. $$
Therefore, combining $I^\ast_{1}, I^\ast_{2}, I^\ast_{3}$ and $I^\ast_{4}$ we get,
\begin{align*}
    \left\|\pi_{kl}(x) - \mathrm{1}_{(a_{kl}, a_{k(l+1)}]}(x)\right\|_1 = I^\ast_{1} + I^\ast_{2} + I^\ast_{3} + I^\ast_{4} \leq \frac{\epsilon}{3\|B_{0k}\|KK^\prime}.
\end{align*}
This implies that given $\pi$, with positive probability the following holds, 
\begin{align}\label{eq: pi_ind}
  T_2 = \left\|  \sum_{k = 1}^K \sum_{l = 1}^{K^\prime} \pi_{kl}(x) B_{0kl} - \sum_{k = 1}^K \sum_{l = 1}^{K^\prime} \mathrm{1}_{(a_{kl}, a_{k(l+1)}]}(x)B_{0kl} \right\|_1 \leq \frac{\epsilon}{3}.
\end{align}

\underline{Bound $T_3$}. Note that the prior of $B_{kl}$ is absolutely continuous on $\mathbb{S}_{p}$. Therefore, we have,
$$\mathbbm{P}\left( \left\| B_{0kl} -  B_{kl} \right\|_1 \leq \frac{\epsilon}{3}\right)>0.$$ Also we have that,
 \begin{align}
     T_3 = \left\| \sum_{k = 1}^{K} \sum_{l = 1}^{K^\prime} \pi_{kl} (\cdot) B_{0kl} - \sum_{k = 1}^{K} \sum_{l = 1}^{K^\prime} \pi_{kl} (\cdot) B_{kl} \right\|_1 \leq 
     \sum_{k = 1}^{K} \sum_{l = 1}^{K^\prime} \left\| B_{0kl} -  B_{kl} \right\|_1  \leq \frac{\epsilon}{3}.
 \end{align}
This implies, for given $\pi$ we have, 
\begin{align}\label{eq: B0k_Bk}
    \mathbbm{P}\left(T_3 \leq \frac{\epsilon}{3} \bigm \vert \pi \right) > 0.
\end{align}
Combining Equations~\eqref{eq: B0_B0k}, \eqref{eq: pi_ind} and \eqref{eq: B0k_Bk} we obtain that given $\pi$, with positive probability the following holds,
\begin{align}
        \left\| B_0(\cdot) - \sum_{k = 1}^{K} \sum_{l = 1}^{K^\prime} \pi_{kl} (\cdot) B_{kl} \right\|_1 \leq T_1 + T_2 + T_3 \leq \epsilon.
\end{align}
Finally by marginalizing out $\pi$, we obtain that,
Hence we have,
\begin{align*}
    \mathbbm{P} \left(\left\|B_0(\cdot) - \sum_{j = 1}^{K^*} \pi_{j} (\cdot) B_{j} \right\|_1 \leq \epsilon \right) >  0.
\end{align*}
which concludes the proof.

\item [$(ii)$]  The proof of this  part follows in the exact similar way as part $(i)$ of the theorem and Lemma~\ref{lemma:HolderCont} where we choose $q = 1$, $B_0 = M_0$, $B_j = M_j$ and absolutely continuous prior on $M_j$ in $\mathbb{R}^{p}$.

\end{proof}

\end{document}